%% file: main.tex
\documentclass[sigconf, nonacm, pdfa]{acmart}

\newcommand\vldbdoi{10.14778/3725688.3725707}
\newcommand\vldbpages{1798 - 1810}
\newcommand\vldbvolume{18}
\newcommand\vldbissue{6}
\newcommand\vldbyear{2025}

\newcommand\vldbtitle{\shorttitle} 
\newcommand\vldbavailabilityurl{https://gitlab.com/alexandra-rogova/neo4j_increasing_value_test}
\newcommand\vldbpagestyle{empty}

\input{packages}

\input{macros}

\title{GQL and SQL/PGQ: Theoretical Models and Expressive Power}

\author{Am\'elie Gheerbrant}
\affiliation{%
  \institution{Universit\'e Paris Cit\'e, CNRS, IRIF}
  \city{Paris}
  \country{France}}
\orcid{0000-0002-8936-9829}
\email{amelie@irif.fr}

\author{Leonid Libkin}
\affiliation{%
  \institution{RelationalAI \& University of Edinburgh}
  \city{Paris \& Edinburgh}
  \country{UK \& France}}
  \orcid{0000-0002-6698-2735}
  \email{l@libk.in}

\author{Liat Peterfreund}
\affiliation{%
  \institution{School of CS, Hebrew University}
  \city{Jerusalem}
  \country{Israel}}
\orcid{0000-0002-4788-0944}
\email{liat.peterfreund@mail.huji.ac.il}

\author{Alexandra Rogova}
\affiliation{%
  \institution{Université Paris Cité, CNRS, IRIF}
  \city{Paris}
  \country{France}}
\orcid{0000-0003-3824-445X}
\email{rogova@irif.fr}

\date{}

\begin{document}

\begin{abstract}
SQL/PGQ and GQL are very recent international standards for querying property graphs: SQL/PGQ specifies how to query relational representations of property graphs in SQL,
while GQL is a standalone language for graph databases. The rapid industrial development of these standards left the academic community trailing in its wake. While digests of the languages have appeared, we do not yet have concise foundational models like relational algebra and calculus for relational databases that enable the formal study of languages, including 
their expressiveness and limitations.
At the same time, work on the next versions of the standards has already begun, to address the perceived limitations of their first versions.

Motivated by this, 
we initiate a formal study of SQL/PGQ and GQL, concentrating on their concise formal model and expressiveness. For the former, we define simple core languages -- Core PGQ and Core GQL -- that capture the essence of the new standards,
are amenable to theoretical analysis, and clarify the difference
between PGQ's bottom up evaluation versus GQL's linear, or pipelined
approach. Equipped with these models, we both confirm the necessity to extend the language to fill in the expressiveness gaps and identify the source of these deficiencies. We complement our theoretical analysis with an experimental study, demonstrating that existing workarounds in full GQL and PGQ are impractical, further underscoring the necessity to correct deficiencies in language design.
\end{abstract}


\maketitle

\pagestyle{\vldbpagestyle}
\begingroup\small\noindent\raggedright\textbf{PVLDB Reference Format:}\\
Amelie Gheerbrant, Leonid Libkin, Liat Peterfreund, and Alexandra Rogova. 
\vldbtitle. PVLDB, \vldbvolume(\vldbissue): \vldbpages, \vldbyear.\\
\href{https://doi.org/\vldbdoi}{doi:\vldbdoi}
\endgroup
\begingroup
\renewcommand\thefootnote{}\footnote{\noindent
This work is licensed under the Creative Commons BY-NC-ND 4.0 International License. Visit \url{https://creativecommons.org/licenses/by-nc-nd/4.0/} to view a copy of this license. For any use beyond those covered by this license, obtain permission by emailing \href{mailto:info@vldb.org}{info@vldb.org}. Copyright is held by the owner/author(s). Publication rights licensed to the VLDB Endowment. \\
\raggedright Proceedings of the VLDB Endowment, Vol. \vldbvolume, No. \vldbissue\ %
ISSN 2150-8097. \\
\href{https://doi.org/\vldbdoi}{doi:\vldbdoi} \\
}\addtocounter{footnote}{-1}\endgroup

\ifdefempty{\vldbavailabilityurl}{}{
\vspace{.3cm}
\begingroup\small\noindent\raggedright\textbf{PVLDB Artifact Availability:}\\
The source code, data, and/or other artifacts have been made available at \url{\vldbavailabilityurl}.
\endgroup
}

\section{Introduction}
\input{sec-introduction}

\section{GQL and SQL/PGQ by examples}\label{sec:ex}
\input{sec-gql-ex}

\section{Pattern Matching: Turning Property Graphs into Relations}\label{sec:pm}
\input{sec-pm}

\section{SQL/PGQ: Theoretical Abstractions}\label{sec:pgq}
\input{sec-pgq}

\section{GQL: Theoretical Abstractions}\label{sec:gql}
\input{sec-gql}

\section{Limitations of Patterns: a theoretical investigation}\label{sec:expressiveness}
\input{sec-express}

\section{Limitations of Patterns: an experimental investigation}\label{sec:experiments}
\input{sec-experiments}

\section{What we learned about language design problems}\label{sec:concl}
\input{sec-conclusions}

\section*{Acknowledgments}
This research was supported by ANR-21-CE48-0015 Project VeriGraph, a grant from RelationalAI to IRIF, Poland's National Science Centre grant 2018/30/E/ST6/00042, and 
Israel Science Foundation grant 2355/24. 

\OMIT{
\section*{Acknowledgments}
This research was supported by ANR-21-CE48-0015 Project VeriGraph (Leonid Libkin), a grant from RelationalAI to IRIF, Poland's National Science Centre grant 2018/30/E/ST6/00042 (Alexandra Rogova) and 
Israel Science Foundation 2355/24 (Liat Peterfreund).
}

\bibliographystyle{abbrv}
\balance
\bibliography{references}

\newpage
\onecolumn
\appendix
\input{appendix}

\end{document}

%% file: packages.tex
\usepackage[utf8]{inputenc}

\usepackage{amsmath, amssymb, amsthm, amsfonts}
\usepackage{stmaryrd}
\usepackage{bbding}
\usepackage{comment}

\usepackage{mathtools}
\usepackage{semantic}
\usepackage{xspace}

\usepackage{multirow,array}

\usepackage{appendix}
\usepackage{hhline}
\usepackage{booktabs}

\usepackage{tikz}
\usetikzlibrary {positioning}
\usepackage{blindtext}
\usepackage{subcaption}

\usepackage{balance}

%% file: macros.tex
\usepackage[final]{listings}

\usepackage{color}
\definecolor{gray}{RGB}{26,30,35}
\definecolor{darkblue}{RGB}{68,89,197}
\definecolor{darkred}{rgb}{0.45,0,0}
\definecolor{darkgreen}{RGB}{0, 156, 0}
\definecolor{darkpurple}{RGB}{120, 0, 180}

\newcommand{\OMIT}[1]{}
\newcommand{\mcomment}[2]{{\color{blue}\textbf{(#1)}}\footnote{\textbf{#1:} #2}}
\newcommand{\liat}[1]{\mcomment{Liat}{#1}} 
\newcommand{\alex}[1]{\mcomment{Alexandra}{#1}}

\newcommand{\todo}[1]{{\color{blue}{\textbf{TODO:} #1}}}
\newcommand{\new}[1]{{\color{black}#1}}

\renewcommand{\phi}{\varphi}   

\newcommand{\pat}{\psi}
\newcommand{\df}{:=}
\newcommand{\query}{\mathcal Q}

\newcommand{\linquery}{\mathbf L}

\newcommand{\db}{\mathcal{D}}
\newcommand{\gdb}{G}

\newcommand{\Univ}{\mathbf U}
\newcommand{\tup}{\mu}
\newcommand{\domfunc}[1]{\mathsf{dom}(#1)}

\newcommand{\join}{\bowtie}

\newcommand{\tupunion}{\join}

\newcommand{\arity}[1]{\mathsf{arity}(#1)}
\newcommand{\attr}[1]{\mathsf{attr}(#1)}
\newcommand{\rel}{\mathbf R}

\newcommand{\dom}[1]{\mathsf{dom}(#1)}

\newcommand{\true}{\textbf{true}}

\newcommand{\lbl}{\mathsf{lab}}
\newcommand{\src}{\mathsf{src}}
\newcommand{\tgt}{\mathsf{tgt}}
\def\endpoints{\textsf{endpoints}}
\newcommand{\prop}{\mathsf{prop}}

\renewcommand{\path}[1]{\mathsf{path}(#1)}
\newcommand{\concat}{\cdot}

\newcommand{\sch}[1]{\mathsf{FV}\left(#1\right)}
\newcommand{\schb}[1]{\mathsf{FV}\Big(#1\Big)}

\newcommand{\Valueset}{\mathsf{Values}}
\newcommand{\Constset}{\mathsf{Const}}
\newcommand{\Nodeset}{\mathsf{N}}

\newcommand{\Edgeset}{\mathsf{E}}

\newcommand{\Set}[1]{\left\{#1\right\}} 
\newcommand{\labelset}{\ensuremath{\mathcal{L}}}
\newcommand{\keyset}{\ensuremath{\mathcal{K}}}
\newcommand{\constset}{\ensuremath{\textsf{Const}}}

\newcommand{\pathval}{\textsf{path}}
\newcommand{\PP}{\ensuremath{\mathsf{Paths}}}
\newcommand{\Paths}{\mathsf{Paths}}
\newcommand{\pathlen}{\textsf{len}}

\newcommand{\Vars}{\mathsf{Vars}}
\newcommand{\len}{\mathsf{len}}
\newcommand{\lang}[1]{\mathcal L(#1)}

\newcommand{\sem}[1]{\left\llbracket#1\right\rrbracket}

\newcommand{\op}{\circ}

\newtheorem{theorem}{Theorem}[section]
\newtheorem{lemma}[theorem]{Lemma}

\theoremstyle{definition}
\newtheorem{definition}{Definition}[section]

\newcommand{\multipat}{\rho}

\newcommand{\globalsch}{\mathcal S}
\newcommand{\localsch}{\mathbf S}

\newcommand{\RA}{\mathsf{RA}}
\newcommand{\LRA}{\mathsf{LCRA}}
\newcommand{\sLRA}{\mathsf{sLCRA}}

\newcommand{\globalschpat}{\mathsf{Pat}}

\newcommand{\first}[1]{\mathsf{first}(#1)}
\newcommand{\last}[1]{\mathsf{last}(#1)}

\newcommand{\qdisjnodes}{\ensuremath{Q^{\Nodeset}_{\neq}}}
\newcommand{\qdisjedges}{\ensuremath{Q^{\Edgeset}_{\neq}}}

\newcommand{\expressible}{{\color{ForestGreen}\checkmark}}
\newcommand{\notexpressible}{{\color{red}\bigtimes}}

\newcommand{\qNodeLessThan}{\ensuremath{Q^{\Nodeset}_{\uparrow}}}
\newcommand{\qEdgeLessThan}{\ensuremath{Q^{\Edgeset}_{\uparrow}}}
\newcommand{\ua}{\uparrow}
\newcommand{\da}{\downarrow}
\newcommand{\qNode}{\ensuremath{Q^{\Nodeset}}}
\newcommand{\qEdge}{\ensuremath{Q^{\Edgeset}}}

\newcommand{\edb}{\mathcal E}
\newcommand{\idb}{\mathcal I}
\newcommand{\dlrules}{\Phi}
\newcommand{\dlrule}{\phi}
\newcommand{\dloutrule}{\mathtt{Out}}

 \newenvironment{repeatresult}[2]
 {\vskip0.5em\par\textsc{#1 #2.}\em}
 {\vskip1em}

\newcommand{\return}{\Omega} 
\newcommand{\returnEl}{\omega}

\newcommand{\concatto}{\odot}

\usepackage{textcomp}
\newcommand{\ttpcr}{\renewcommand{\ttdefault}{pcr}\ttfamily}
\newcommand{\sqlkw}[1]{%
	\text{\normalfont\small\ttpcr\bfseries\color{darkblue} #1}}

\lstnewenvironment{sql}[1][]{%
  \lstset{%
    language=SQL,%
    basicstyle=\small\color{gray}\ttpcr,%
    keywordstyle=\bfseries\color{darkblue},%
    morekeywords={REFERENCES,IS,COLUMNS,GRAPH_TABLE,COLUMNS,CREATE, PROPERTY,GRAPH,VERTEX,EDGE,TABLES,PROPERTIES,SOURCE,DESTINATION,LABEL,WITH},%
    deletekeywords={YEAR,ADD},%
  }%
  \lstset{#1}%
}{}

\lstnewenvironment{gql}[1][]{%
  \lstset{%
    language=SQL,%
    basicstyle=\small\color{gray}\ttpcr,%
    keywordstyle=\bfseries\color{darkblue},%
    morekeywords={REFERENCES,IS,COLUMNS,GRAPH_TABLE,COLUMNS,CREATE, PROPERTY,GRAPH,VERTEX,EDGE,TABLES,PROPERTIES,SOURCE,DESTINATION,MATCH,FILTER,RETURN,WITH},%
    deletekeywords={YEAR},%
  }%
  \lstset{#1}%
}{}

\newcommand{\nlog}{\ensuremath{\textsc{NLogSpace}}}
\newcommand{\dlog}{\ensuremath{\textsc{DLogSpace}}}

\usepackage{tikz}
\usetikzlibrary{calc}

\newcommand{\currentsidemargin}{%
  \ifodd\value{page}%
    \oddsidemargin%
  \else%
    \evensidemargin%
  \fi%
}

\newlength{\whatsleft}

%% file: sec-introduction.tex
\newcommand{\graphfig}{\begin{figure*}
    \centering
    \includegraphics[width=0.7\textwidth]{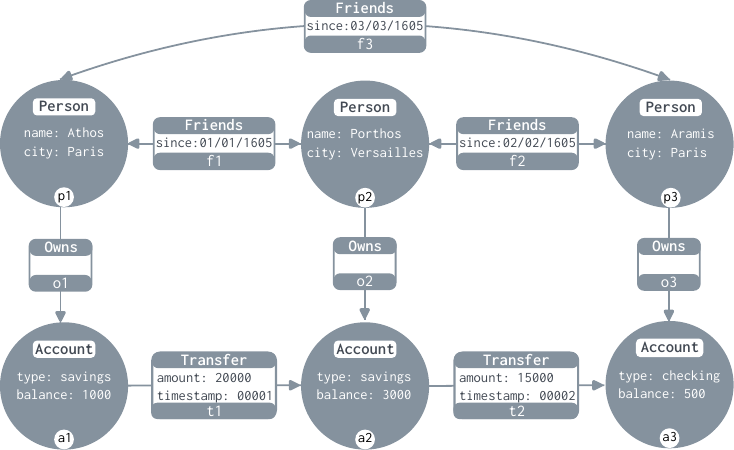}
\caption{A labeled property graph}
\label{fig:example-graph}
\end{figure*}}

In the past two years, ISO published two new international
standards. SQL/PGQ, released in 2023, as Part 16 of the SQL standard, provides
a mechanism for representing and querying property graphs in
relational databases using SQL. GQL, released in 2024, is a standalone
graph query language that does not rely on a relational
representation of graphs. SQL/PGQ and GQL are developed in the same
ISO committee that has been maintaining and enhancing the SQL standard
for decades. These developments reflect the interest of
relational vendors in implementing graph extensions and the emergence
of a new native graph database industry.

There is however a notable difference between this standardization
effort and that of SQL as it happened in the 1980s. Before SQL was
designed, strong theoretical and practical foundations of relational
databases had been developed. Equivalence of relational algebra and
calculus (first-order logic) had been known; these provided very clean
abstractions of relational languages that served as the basis of
relational database theory. This has had a profound impact on both the
theory and the practice of database systems.
\new{For example,
a question of central importance in the early days of relational
databases was expressiveness of query languages, with Fagin, Aho,
Ullman, Gaifman, and Vardi independently showing that relational
calculus cannot define the transitive closure of a relation
\cite{GaifmanVardi85,Fagin75,AhoUllman79}. This led  to
subsequent development of Datalog and culminated in the inclusion of linear
recursive queries by the SQL standard in 1999 \cite{Melton-SQL99}.
}

When it comes to graphs,
the adoption of languages
by industry runs well ahead of the development of their foundational
underpinnings. For many years, work on graph queries concentrated on
regular path queries (RPQs) \cite{RPQ} and their multiple
extensions (e.g.,  \cite{crpq,C2RPQ,BarceloLLW-tods12,B13,surveyChile}). 
\new{These however assume a much simplified model in which neither nodes nor edges have properties, unlike {\em property graphs} preferred in industry. A step in that direction is the model of {\em data graphs} in which nodes but not edges carry data \cite{gxpath-jacm}. To this day, these models are the basis of work on query language expressiveness \cite{anthony-express,wim-trichotomy}. But  we will see soon that from the point of view of GQL and SQL/PGQ, a full property graph model makes a huge difference.  
Existing theoretical languages such as  GXPath \cite{gxpath-jacm}, 
regular queries \cite{RRV17},  STRUQL \cite{struql}, while having a direct influence on the design of GQL \cite{gql-influence}, are not a good reflection of it. 
}

\new{
At the same time, extensive discussions in the standards committees are already under way to
identify new features of SQL/PGQ and GQL, based on their {\em perceived}, rather than
proved, shortcomings \cite{tobias,fred}.  
These considerations motivate our main contributions: }
\begin{enumerate}
\item We provide concise formal models of SQL/PGQ and GQL;
\item We confirm the intuition that certain queries of interest to
customers of graph databases cannot be
expressed by SQL/PGQ and GQL patterns;
\item We compare GQL with recursive SQL and show that the latter is
  more powerful; and
\item \new{We show experimentally that methods currently allowed in SQL/PGQ and GQL to overcome these limitations are impractical, even for small-sized graphs.}
\end{enumerate}

\OMIT{Of course, when we say ``cannot be expressed'' or ``is more
powerful'', these refer to the theoretical models of the
languages, in line with similar results for SQL and relational
calculus/algebra.
}

We now elaborate more on these goals.

\paragraph{\underline{Formal models of GQL and SQL/PGQ}}
The workhorse of graph query languages is {\em pattern
matching}. In Cypher, PGQ, GQL and others,
pattern matching turns a graph into a table; the remaining operations
of the language manipulate that table.
In fact in GQL and SQL/PGQ pattern matching is {\em identical}; it is
only how its results are processed that is different.
In SQL/PGQ, a graph is a view defined on a relational
database. The result of pattern matching 
is simply a table in the \sqlkw{FROM} clause of a SQL query
that may use other relations in the database.
GQL, on the other hand, is oblivious to how a graph is stored, and
the table resulting from pattern matching is manipulated by a sequence
of operators that modify it, in an imperative style that is referred
to as ``linear composition''.

To produce a formal model of these languages, we use the relationship
between SQL and relational calculus/algebra as the guiding
principle. SQL has a multitude of features: bag semantics,
nulls, arithmetic operations, aggregate functions, outerjoins, complex
datatypes such as arrays.  Some of these break the theoretical model
of First Normal Form (1NF) relations which are {\em sets} of tuples
of {\em atomic values}.  Relational algebra and
calculus get rid of the extra baggage that SQL is forced to add to be
usable in practice, and yet provide a simple core over sets of tuples
of atomic values.

\sloppy
Although some mathematical abstractions of GQL and SQL/PGQ
exist \cite{icdt23,pods23}, they are not yet at the same level as RA
or relational calculus, in terms of their simplicity and utility in formally proving
results. \new{For example, their formalizations of pattern matching
  come with a complex typing system and the use of conditional and
  group variables that create nulls and set-valued attributes.  
To achieve our goal of creating a simple and usable abstraction, we need to:
simplify the model of pattern matching, to 
avoid outputting non-1NF relations, and 
to formalize the linear composition of Cypher and GQL. This allows us
to define two languages
-- Core PGQ and Core GQL -- as RA or linear composition of
  relational operators on top
of pattern matching outputs. 
}

\paragraph{\underline{Limitations of pattern matching}}

Many pattern matching tasks require iterating patterns, notably when
we do not know a priori the length of a path we are searching for (a
basic example is reachability in graphs). 
The way pattern matching is designed in Cypher, GQL, and SQL/PGQ,
makes it easy to compare two consecutive iterations of a pattern based
on property values of their nodes, but at the same time it is very
hard to compare property values in {\em edges}.

To illustrate this, assume we have a chain of transfers between
accounts: accounts are nodes, with property values such as ``balance'',
and transfers are edges, with property values such as
``amount'' and ``timestamp'', between accounts. It is very easy to write a query looking
for a chain of transfers such that the \textit{balance} increases in accounts
along the chain. It appears to be very hard or impossible to write a
query looking for a chain of transfers such that the \textit{timestamp}
increases along the path.

Demand for these queries occurs often in practice, which forced
several companies participating in GQL design in ISO to start thinking
of language extensions that will permit such
queries \cite{tobias,fred}. However, before making non-trivial
language enhancements, it would be good to actually know that newly
added features cannot be achieved with what is already in the language.

We show that this is indeed the case.
\new{Equipped with our formalization of the pattern language, we
  prove that it cannot express a large class of queries that analyze
  how property values of edges change along matched paths; of these the {\em
increasing value in edges} query is the simplest. This follows in fact
  from a more general property that is akin to a pumping lemma for
  paths that can be selected by GQL and PGQ patterns.
}


\OMIT{
Furthermore, we use our formalization to clear up an issue related to
Cypher pattern matching. It has been claimed multiple
times (e.g., \cite{cypher,graphdb-book,surveyChile}) that Cypher's
patterns fall short of the full power of RPQs. In fact this was a key
motivation for extending its pattern matching language \cite{cypher}
to that of GQL \cite{sigmod22}. However, this ``folklore'' result has
never been actually proven, not least due to the lack of a suitable
formalization of Cypher. We use our formalization of GQL to define its
fragment that corresponds to the (original design of)
Cypher \cite{cypher}, and then formally show that not all RPQs are
expressible in it. 
}

\paragraph{\underline{GQL vs Recursive SQL vs Datalog}}

When the
ubiquitous CRPQs was introduced in \cite{crpq}, it was shown that a Datalog-like language based on
CRPQs has the power of transitive closure logic \cite{FMT} and
captures the complexity class \nlog\ of problems solved by
nondeterministic Turing machines whose work tape used logarithmic
number of bits in terms of the size of the input (this class is
contained in polynomial time). Since then, \nlog\ is the yardstick
complexity class we compare graph languages to, and Datalog -- that
subsumes the transitive closure logic -- is a
typical language used to understand the power of graph querying.
{
Datalog is also the basis of 
SQL's recursive common table expressions (CTE) introduced
by the
\sqlkw{WITH RECURSIVE} clause. 
In fact, in SQL only {\em linear
Datalog} rules are allowed: in them, the recursively defined
predicate can be used at most once.

Motivated by this, we compare the power of GQL with  recursive SQL
and linear Datalog. 
We show that there are queries that are
expressible in the latter,  
have very low  data complexity, 
and yet are {\em not expressible} in GQL. We explain how
this points out deficiencies of GQL design that will hopefully be
addressed in the future. 

}

\OMIT{
On SQL's side, Datalog recursive rules are represented by recursive
CTEs, using the \sqlkw{WITH RECURSIVE} clause. 
In fact, in SQL only {\em linear
Datalog} rules are allowed: in such rules, the recursively defined
predicate can be used at most once. This class of queries is not only
the basis of recursive SQL but in general plays a special role in many
applications of recursive queries \cite{ketsman-koutris}. 

Motivated by this, we compare the power of GQL with (a) recursive SQL
(b) linear Datalog, and (c) complexity class \nlog. At first it looks
like GQL should be able to capture \nlog\ and linear Datalog (whose
data complexity is in \nlog\ \cite{GottlobP03}). Indeed, reachability (that is easily
expressible by GQL patterns) is complete for \nlog\ under first-order
reductions, and such reductions are expressible in RA.\liat{I'm trying to reformulate this. It is not clear enough to this audience}

Defying this intuition, we show that there are queries that are
expressible in recursive SQL and linear Datalog, have \dlog\ data
complexity, and yet are {\em not expressible} in GQL. We explain how
this points out some deficiencies of GQL design that will hopefully be
addressed in the future. 
}

\new{
\paragraph{\underline{Experimental evaluation}: can graph DBMSs
  overcome these limitations?}
Real-life systems go beyond basic calculi; SQL with recursion and
aggregates is in fact Turing-complete and thus can express all
computable queries. Similarly, GQL, PGQ, Cypher and others have many
tools at their disposal to let users write very powerful queries. In
fact queries such as {\em increasing value in edges} are expressible
in real-life Cypher and PGQ, though in a very convoluted way. Since
the {\em complement} of the query is easily expressible, one can look
for all the paths and then subtract the complement. Intuitively, such
a way should not work in practice as the number of paths in a graph
grows exponentially.

This is precisely what we confirm using three different
implementation of graph database queries: Neo4j and Memgraph for
Cypher, and DuckDB for SQL/PGQ. Native graph systems can handle just a
few dozen nodes before 100\% timeout rate is observed; DuckDB does
marginally better as it computes fewer paths to start with. 
It is important to note that this is not a critique of the systems tested, as the indirect method of computing the query inherently requires an exponential number of paths. Instead, this confirms that no workaround can bypass the inherent expressibility limitations.
}

\smallskip

Since inexpressibility results are commonly used to identify deficiencies in
language design, and serve as a motivation to increase language
expressiveness,
we shall at the end of the paper discuss what these results tell
us in terms of new features of GQL and SQL/PGQ that will be required.

\graphfig

\subsubsection*{Related work}
While industry is dominated by property graphs (Neo4j \cite{cypher}, Oracle \cite{PGQL}, Amazon \cite{Neptune}, TigerGraph \cite{tigergraph-sigmod},
etc), much of academic literature still works with the model of
labeled graphs and query languages based on RPQs, with some
exceptions \cite{gxpath-jacm,ABFF18,defensive-citation-one}. 
However, the rare models focused on property graphs appeared before the new standards became available, and their analyses of
expressiveness and language features do not apply to GQL and
SQL/PGQ. The first commercial language 
for property graphs was Cypher, and it was fully formalized
in \cite{cypher}. As GQL and SQL/PGQ were being developed, a few academic papers
appeared. For example, \cite{sigmod22} gave an overview of their pattern matching
facilities, which was then further analyzed
in \cite{pods23}, and \cite{duckpgq} provided a description of an early implementation of SQL/PGQ. 

In \cite{icdt23}, a digest of GQL suitable for the
research community was presented. While a huge improvement compared to
the actual standard from the point of view of clarity, the
presentation of \cite{icdt23} replaced 500 pages of the text of the
Standard (notoriously hard to read) by a one-page long definition of
the syntax, followed by a four-pages long definition of the
semantics. 
It achieved a two orders of magnitude reduction in the size of
the definition of the language,  
but 5 pages is still way too long for ``Definition 1''. 
\new{The language of \cite{pods23} is closer to our goal, but it is not well-suited to the level reasoning we require here as it still contains still too much of the GQL and PGQ baggage, for example non-1NF relations, and a complex type
system for singleton, conditional, and group variables that we avoid here entirely and replace by the very familiar definition of free variables.  }

\OMIT{The pattern
matching defined in \cite{pods23} also had too much of the GQL and PGQ
baggage, producing non-1NF relations, and it used a complex type
system for singleton, conditional, and group variables. That type
system is rather complex to describe and it owes to the design
decisions that had more to do with getting the new standards over the
line than producing a perfect language design. As such it was not really suitable for a formal analysis. }

Another element missing from the literature is a proper
investigation of linear composition. Initially introduced in Cypher,
it was then  adopted in a purely relational language PRQL \cite{prql}
(positioned as ``pipelined'' alternative to SQL), embraced by
GQL, and influenced the designed of the piped SQL syntax \cite{google-pipes}; however formal analyses of this way of building complex queries are
lacking.

%% file: sec-gql-ex.tex
We illustrate GQL and SQL/PGQ capabilities using the graph from
Figure \ref{fig:example-graph} and the following money-laundering query: 
{\em Find a pair of friends in the same city who transfer money to each
	other via a common friend who lives elsewhere}. 
Notice that we assume that friendship is a symmetric relation.

In GQL, it will be expressed as query: 

\begin{gql}
	MATCH (x)-[:Friends]->(y)-[:Friends]->(z)
	-[:Friends]->(x), (x)-[:Owns]->(acc_x),
	(y)-[:Owns]->(acc_y), (z)-[:Owns]->(acc_z), 
	(acc_x)-[t1:Transfer]->(acc_z)
	-[t2:Transfer]->(acc_y)
	FILTER (y.city = x.city) AND (x.city<>z.city)
	AND (t2.amount < t1.amount)
	RETURN x.name AS name1, y.name AS name2
\end{gql}

In SQL/PGQ, a graph is a view of a tabular schema. The graph from
Figure \ref{fig:example-graph} can be represented by the following set of tables: 
\begin{itemize}
	\item \verb+Person(p_id,name,city)+ for people,
	\item \verb+Acc(a_id,type)+ for accounts,
	\item \verb+Friend(p_id1,p_id2,since)+  for friendships,
	\item \verb+Owns(a_id,p_id)+ for ownership and 
	\item \verb+Transfer(t_id,a_from,a_to,amount)+ for transfers.
\end{itemize}
The property graph view is then defined by a $\sqlkw{CREATE}$ statement, part of
which is shown below:

\begin{sql}
	CREATE PROPERTY GRAPH Interpol1625 (
	VERTEX TABLES 
	Acc KEY (a_id) LABEL Account PROPERTIES (type)
	....
	EDGE TABLES 
	Transfer KEY (t_id)
	SOURCE KEY (a_from) REFERENCES Acc 
	DESTINATION KEY (a_to) REFERENCES Acc
	LABEL Transfer PROPERTIES (amount)
	.... ) 
\end{sql}

This view defines nodes (vertices) and edges of the graph, specifies
endpoints of edges, and defines their labels and properties. 
We can then query it, using pattern matching to
create a subquery:

\begin{sql}
	SELECT T.name_x AS name1, T.name_y AS name2
	FROM Interpol1625 GRAPH_TABLE (
	MATCH (x)-[:Friends]->(y)-[:Friends]->(z)
	-[:Friends]->(x),  (x)-[:Owns]->(acc_x),
	(y)-[:Owns]->(acc_y), (z)-[:Owns]->(acc_z), 
	(acc_x)-[t1:Transfer]->(acc_z)
	-[t2:Transfer]->(acc_y)
	COLUMNS x.city AS city_x, y.city AS city_y,
	z.city AS city_z,
	x.name AS name_x, y.name AS name_y, 
	t1.amount AS amnt1, t2.amount AS amnt2 
	) AS T
	WHERE T.city_x=T.city_y 
    AND T.city_x <> T.city_z
	AND T.amount1 > T.amount2 
\end{sql}

Note that in GQL, a sequence of operators can continue {\em after} the
$\sqlkw{RETURN}$ clause. For example, if we want to find large transfers between
the two potential offenders we could simply continue the first GQL query with extra clauses:
\begin{gql}
	MATCH (u WHERE u.name=name1)
	-[t:Transfer]->
	(v WHERE v.name=v2)
	FILTER t.amount > 100000
	RETURN t.amount AS big_amount
\end{gql}

This is what is referred to as {\em linear composition}: we can simply
add clauses to the already existing query which apply new
operations to the result of already processed clauses. 

In SQL/PGQ, such an operation is also possible, though perhaps a bit
more cumbersome as we would need to put the above PGQ query as a subquery in \sqlkw{FROM} 
and create another subquery for the second match, then join them on
{\tt name1} and {\tt name2}.

%% file: sec-pm.tex
We define property graphs and pattern matching, the key component of GQL  and SQL/PGQ, that extracts relations from  graphs.

\subsection{Property Graphs}

We use the standard definition (cf.~\cite{pods23}), 
and only consider directed edges.
Assume pairwise
disjoint countable sets $\labelset$ of labels, $\keyset$ of keys,
$\Constset$ of constants, $\Nodeset$ of node ids, and $\Edgeset$ of
edge ids.

\begin{definition}[Property Graph]
	\label{pg-def}
	A property graph  is a tuple 
	$
	\gdb = \langle N, E, \lbl, \src, \tgt, \prop\rangle
	$
	where
	\begin{itemize}
		\item $N \subset \Nodeset$ is a finite set of node ids used in $\gdb$;
		\item $E \subset \Edgeset$ is a finite set of directed edge ids used in $\gdb$;
		\item $\lbl:  N  \cup  E \to 2^{\labelset}$  
		is a labeling function that associates with every node or edge id a (possibly empty) finite set
		of labels from $\labelset$;
		\item $\src, \tgt: E \to N$ define source and target of an edge;
		\item $\prop: (N  \cup  E) \times \keyset \to \constset$ is a partial
		function that associates a property value
        with a node/edge id and a key. 
	\end{itemize}
\end{definition}

A \emph{path} in $G$ is an alternating sequence $u_0 e_1 u_1 e_2 \cdots
e_{n} u_n$, for $n \geq 0$, of nodes and edges that starts and ends with a node and so
that each edge $e_i$ connects the nodes $u_{i-1}$ and $u_i$ for $i
\leq n$.
More precisely, for each $i \leq n$, either $\src(e_i)=u_{i-1}$ and
$\tgt(e_i)=u_i$ (a {\em forward edge}), or  $\src(e_i)=u_{i}$ and
$\tgt(e_i)=u_{i-1}$ (a {\em backward edge}). 
Note that $n=0$ is possible, in which case the path consists of a
single node $u_0$. We shall explicitly spell out paths as
$\pathval(u_0,e_1,u_1,\cdots, e_n,u_n)$.

Two paths $p = \pathval(u_0,e_0,\ldots,u_k)$ and $p' =
\pathval(u'_0,e'_0,\ldots,u'_j)$ 
\emph{concatenate}, written as $p \concatto p'$, 
if $u_k = u'_0$, in
which case their \emph{concatenation} $p \concat p'$ is defined as
$\pathval(u_0,e_0,\ldots,u_k,e'_0,\ldots,u'_j)$. %
Note that a single-node path is a unit of
concatenation: $p \concat \pathval(u)$
is defined iff $u=u_k$ and is equal to $p$.

\subsection{Pattern Matching}


Pattern matching is the key component of graph query languages.
As already mentioned, an early abstraction of GQL and PGQ patterns was
given in \cite{pods23}, but it retained too much of the baggage of the
actual language (non-1NF outputs, nulls, a complex type system) for
language analysis.
Thus, here we  refine the definitions from~\cite{pods23} to capture the core concepts of pattern matching, similarly to relational algebra for relational databases. To this end,
we fix an infinite set $\Vars$ of
variables and define Core GQL and Core PGQ pattern matching as follows:
\begin{center}
$
	\pat \ \ \df \ \  (x) \ \mid \
	\overset{x}{\rightarrow} \ \mid \
	\overset{x}{\leftarrow} \ \mid \
	\pat_1\, \pat_2  \ \mid\  
	\pat^{n..m} \ \mid\ 
	\pat\langle\theta \rangle \ \mid \ 
	\pat_1 + \pat_2
	$
	\end{center}
	where
	\begin{itemize}
\item $x \in \Vars$ and $0 \leq n \leq m \leq \infty$;
\item variables $x$ in node and edge patterns $(x)$,
$\overset{x}{\rightarrow}$, and $\overset{x}{\leftarrow}$ are optional,
\item $\pat\langle\theta \rangle$ is a conditional pattern, and
conditions  are given by $\theta, \theta' \ \df \ x.k=x'.k' \mid
x.k<x'.k' \mid \ell(x) \mid 
\theta \vee \theta'
\mid
\theta \wedge \theta'
\mid \neg \theta$
where $x,x'\in\Vars$ and $k,k'\in\keyset$;
\item $\pat_1+\pat_2$ is only defined when their sets of free
variables $ \sch{\pat_1}$
and $\sch{\pat_2}$ are equal.
\end{itemize}

The sets of free variables are defined as follows:
\begin{itemize}
\item $\schb{(x)} = \sch{\overset{x}{\rightarrow}} =  \sch{
	\overset{x}{\leftarrow}} \ \df \ \{x\}$;

\item   $\sch{\pat_1 + \pat_2} \ \df\  \sch{\pat_1}$
\item $\sch{ \pat_1\, \pat_2 } \ \df\ \sch{\pat_1}\cup\sch{ \pat_2 }$
\item $\sch{   \pat^{n..m} } \ \df\ \emptyset$
\item $\sch{\pat\langle\theta \rangle}  \ \df\ \sch{\pat}$
\end{itemize}

A pattern produces an output that consists of graph elements and their
properties. Such an 
output $\return$ is a (possibly empty)
tuple whose elements are either variables $x$ or properties $x.k$. 
A {\em pattern with output}, \new{or, pattern, for simplicity,} is an expression $\pat_\return$ such
that every variable present in $\return$ is in $\sch{\pat}$.  

\paragraph{Correspondence with Cypher and GQL}
For the reader familiar with Cypher and/or GQL, we explain how our
fomalization compares with these languages' patterns.
\begin{itemize}
\item $(x)$ is a node pattern that binds the variable $x$ to a
node;
\item   $\overset{x}{\rightarrow}$ and $\overset{x}{\leftarrow}$
are forward edge and backward edge patterns, that also bind $x$ to
the matched edge;
\item  $\pat_1\, \pat_2$ is the concatenation of patterns, 
\item $\pat^{n..m}$ is the repetition of $\pat$ between $n$ and $m$ times
(with a possibility of $m=\infty$)
\item  $\pat\langle\theta \rangle$
corresponds to $\sqlkw{WHERE}$ in patterns, conditions involve
(in)equalities between property values, checking for labels, and their Boolean combinations;
\item  $\pat_1 + \pat_2$ is the union of
patterns;
\item $\pat_\return$ corresponds to
the output forming clauses \sqlkw{RETURN} of Cypher and GQL
and \sqlkw{COLUMNS} of SQL/PGQ, with $\return$ listing
the attributes of returned relations. 
\end{itemize}

\paragraph{Semantics}
To define the semantics, 
we use set $\Valueset$ which is  the union of
$\Constset \cup 
\Nodeset \cup 
\Edgeset$. That is, its elements are node and edge ids, or values of
properties, i.e., precisely the elements that can appear as outputs of
patterns.

The \emph{semantics of a path pattern} $\pat$, with respect to a graph $\gdb$,
is a set of pairs $(p,\mu)$
where $p$ is  a path and $\mu$ is a mapping $\sch{\pat}\rightarrow
\Valueset$.
Recall that we write $\mu_\emptyset$ for the unique empty mapping 
with $\dom{\mu}=\emptyset$.

For the {semantics of path patterns with output} $\pat_\return$
we define $\mu_\return:\return \to \Valueset$ as the projection of
$\mu$ on $\return$:   
\[\mu_\return (\returnEl) \df
\begin{cases} 
\mu(x) & \text{if } \returnEl = x\in \Vars \\
\prop(\mu(x),k)  & \text{if } \returnEl = x.k\,. 
\end{cases}\]
Full definitions are presented in Figure~\ref{fig:sempathpattern}.
For node and edge patterns with no variables, the mapping part of the
semantics changes to $\mu_\emptyset$. 
The satisfaction of a condition $\theta$ by a mapping $\mu$, written
$\mu\models\theta$, is defined as follows: $\mu\models x.k =
x'.k'$ if 
both $\prop(\mu(x),k)$ and $\prop(\mu(x'),k')$ are both defined and
are equal (and likewise for $<$), and
$\mu\models\ell(x)$ if $\ell\in\lbl(\mu(x))$. It is then extended to
Boolean connectives $\wedge, \vee, \neg$ in the standard way. 

\newcommand{\patsem}
{
$\begin{array}{rl}
	\sem{(x)}_G & 
	\df \Set{(\path{n},\{x\mapsto n\}) \mid n\in \Nodeset}
	\\
	\sem{\overset{x}{\rightarrow}}_G &
	\df 
	\left\{
	(\path{n_1,e,n_2},\{x\mapsto e\}) 
	\ \middle|\
	e\in \Edgeset, \ \src(e)=n_1, \ \tgt(e)= n_2
	\right\}
	\\
	\sem{\overset{x}{\leftarrow}}_G &
	\df  \left\{(\path{n_2,e,n_1},\{x\mapsto e\}) 
	\ \middle|\
	e\in \Edgeset, \ 
	\src{(e)}=n_1, \ \tgt{(e)}= n_2
	\right\}
	\\
	\sem{\pat_1 + \pat_2}_G & \df
	\sem{\pat_1}_G \cup  \sem{\pat_2}_G\\  
	\sem{\pat_1 \,\pat_2}_{G} & \df 
	\left\{ ({p_1\concat p_2}, \mu_1\tupunion \mu_2 ) 
	\ \middle| \ (p_1,\mu_1)\in\sem{\pat_1}_{G},
	\ (p_2,\mu_2)\in\sem{\pat_2}_{G},\
	\mu_1\sim \mu_2, \ p_1 \concatto p_2
	\right\}
	\\
	\sem{\pat \langle\theta \rangle}_G &\df
	\Set{(p,\mu)\in\sem{\pi}_G \ \mid\  
		\mu \models \theta}  \\
	\OMIT{
		& \hspace*{1cm}\mu \models x.k < y.k' \text{ if }
		\mu(x).k < \mu(y).k'
		\\
		& \hspace*{1cm}\mu \models x.k = y.k' \text{ if }
		\mu(x).k=\mu(y).k'
		\\
		&\mu \models \ell(x) \text{ if }
		\ell \in \lbl({\mu(x)})
		\\
		&\mu \models \theta \wedge \theta' \text{ if }
		\mu\models \theta \text{ and } \mu\models \theta'
		\\
		&\mu \models \neg \theta  \text{ if }
		\mu\not \models \theta
	} 
	\sem{\pat^{n..m}}_G & \df \displaystyle{\bigcup_{i=n}^m}  \sem{\pat}^i_G \text{ where }\\
	& \qquad\qquad \sem{\pat}^0_G  \df \left\{(\path{n}, \mu_{\emptyset}) \ \big|\  
	n\in\Nodeset \right\}\\
	& \qquad\qquad {\sem{\pat}}^{n}_{G}  \df  
	\left\{ (p_{1} \cdots p_{n}, \mu_{\emptyset} ) 
	\ | \
	\exists \mu_{1},\ldots,\mu_{n}:
	(p_i,\mu_i)\in{\sem{\pat}}_{G} \text{ and }
	p_{i} \concatto p_{i+1}\text{ for all }i < n \right\},  \, n>0
	\\ & \\
	\sem{\pat_\return}_{\gdb}&\df 
	\left\{\mu_\return \ \mid \ 
	\exists p: (p, \mu) \in \sem{\pat}_{\gdb}
	\OMIT{\text{ if } \returnEl=x \text{ for some } x \text{ then } \mu(\returnEl) = \mu'(\returnEl), \\
		\text{otherwise } \returnEl=x.a \text{ for some } x \text{ and some } a \text{ then } \mu(\returnEl) = \mu'(x).a }
	\right\}
\end{array}
$
}

\begin{figure*}
\centering
\patsem
\caption{Semantics of patterns and patterns with output}
\label{fig:sempathpattern}
\end{figure*}

\paragraph{Pattern languages vs GQL and PGQ patterns}
Compared to GQL and PGQ patterns as described
in \cite{sigmod22,icdt23,pods23}, we make some simplifications.
First and foremost, they are to ensure that outputs of pattern
matching are 1NF relations. Similarly to formal models
of SQL by means of relational algebra and calculus over sets of
tuples, we do not include 
bags, nulls, and 
{relations whose entries are not atomic values.}
The differences manifest themselves in four ways. 

First, we use set semantics rather than bag semantics. GQL and PGQ
pattern can return tables with duplicate patterns; we follow their
semantics up to multiplicity.

 Second, in disjunctions $\pat_1 + \pat_2$ we require that free variables
of $\pat_1$ and $\pat_2$ be the same. In GQL this is not the case, but
then for a variable $x \in \sch{\pat_1}-\sch{\pat_2}$, a match for
$\pat_2$ would generate a {\em null} in the $x$ attribute. Following
the 1NF philosophy, we omit features that can generate nulls.

Third,
repeated patterns $\pat^{n..m}$ have no free variables. In GQL and
PGQ, free variables of $\pat$ become {\em group variables} in
$\pat^{n..m}$, leading to both a complex type system \cite{pods23} and
crucially non-flat outputs. Specifically such variables are evaluated
to {\em lists}. \new{For example, in the pattern $()\xrightarrow{x}^{0..\infty}()$, the variable $x$ would be mapped to the list of edge ids traversed by the path. These lists would be} typically represented as values of an \sqlkw{ARRAY}
type in implementations (cf.~\cite{duckpgq}) thus again violating 1NF. 

Fourth, we do not impose any conditions on paths that can be
matched. In GQL and PGQ they can be {\em simple} paths (no repeated
nodes), or {\em trails} (no repeated edges), or {\em shortest}
paths. In GQL, PGQ, and Cypher, paths themselves may be returned,
and such restrictions therefore are necessary to ensure finiteness of
output. Since we can only 
return graph nodes or edges, or their properties, we never have the
problem of infinite outputs, and thus we chose not to overcomplicate
the definition of core languages by deviating from flat tables as
outputs.

In what {\em might} look like a simplification w.r.t.~GQL
and PGQ, we do not have explicit joins of patterns, i.e., $\pat_1,
\pat_2$ \new{with the semantics $\sem{\pat_1,\pat_2}_{\gdb} =
\left\{\big((p_1,p_2),\mu_1 \tupunion \mu_2\big) \mid (p_i,\mu_i)\in
\sem{\pat_i}_\gdb, \  i=1,2\right\}$. This is because such joins are  definable
with RA operations. This simplification is in the same spirit as not including the natural join in the definition of RA, since it can already be expressed with product, selection, and projection.  
}

%% file: sec-pgq.tex
\new{Having defined patterns shared by SQL/PGQ and GQL, we can proceed to
provide a theoretical abstraction of SQL/PGQ. Recall that in PGQ,
the \sqlkw{MATCH} statement is embedded in \sqlkw{FROM}. In fact
results of matches over a graph are simply treated as relations, or
subqueries, over which the usual SQL can be asked. Taking again the
view that our goal is to provide the core abstraction of the language,
on top of which others can be built, we look at the essential core of
relational languages, namely relational algebra. With this in mind, we
define (for now informally):
\begin{itemize}
\item[] \hspace{-2em}{\em Core PGQ = Relational Algebra over pattern matching outputs.}
\end{itemize}
}

To make this definition formal, we define some very standard
concepts. 
We assume an infinite countable set $\globalsch$ of \emph{relation symbols},
such that each $S\in \globalsch$ is associated with
a sequence  $\attr{S}\df A_1, \ldots, A_n$ of attributes for some
$n>0$; here 
attributes $A_i$ come from a countably infinite set $\mathcal A$ of
attributes. By $\arity S$ we mean $n$, and to be explicit about
names of attributes write $S(A_1, \ldots, A_n)$.


Fix an infinite \emph{domain} $\Univ$ of values. 
A \emph{relation} over $S(A_1,\ldots, A_n)$ is a set of \emph{tuples}
$\tup:\{A_1,\ldots, A_n\} \rightarrow \Univ$. The domain of $\tup$ is
denoted $\domfunc{\tup}$, and we often represent tuples as sets of
pairs $(A,\tup(A))$ for $A\in \domfunc{\tup}$. 
A relational database $\db$ 
is a partial function that maps symbols $S\in \globalsch$ to relations $\db(S)$ over $S$ (if $\db$ is clear from the context we refer to $\db(S)$ simply as $S$, by a slight abuse of notation.) 
We say that $\db$ is over its domain $\dom{\db}$.

Let $\tup$ be a tuple and $\pmb{A} \subseteq \dom{\tup}$. We use
$\tup \restriction \pmb{A}$ to denote the restriction of $\tup$ to
$\pmb{A}$, that is, the mapping $\tup'$ with
$\dom{\tup'} = \pmb{A}$ and $\tup'(A) \df \tup(A)$ for every attribute
$A\in \pmb{A}$.  Two tuples $\tup_1,\tup_2$  are \emph{compatible},
denoted by $\tup_1 \sim \tup_2$, if 
$\tup_1(A) = \tup_2(A)$
for every $A\in \dom{\tup_1} \cap \dom{\tup_2}$.
For such compatible tuples  define
$\tup_1 \tupunion \tup_2$ as the mapping $\tup$ with
$\dom{\tup}\df \dom{\tup_1} \cup \dom{\tup_2}$, and
$\tup(A)\df \tup_1(A)$ if $A\in\dom{\tup_1}$ and
$\tup(A)\df \tup_2(A)$ otherwise.


If $A\in\domfunc{\tup}$ then the renaming $\rho_{ A\rightarrow
B}(\tup)$ of $A$ to $B$ is the mapping $\tup'$ with $\domfunc{\tup'}=
(\domfunc{\tup} \setminus\{A \})\cup\{B\}$ where $\tup'(B)\df \tup(A)$
and $\tup'(A')= \tup(A')$ for every other $A'\in\domfunc{\tup'}$.  


\new{
\paragraph{Relational Algebra (RA)}
We  
use a standard presentation of RA. 
Given a schema $\localsch$ which is a finite
subset of $\globalsch$, the expressions $\query$ of 
$\RA(\localsch)$ and selection conditions $\theta$ are defined as 

$$\begin{array}{rcl}
	\query, \query'  & \!\!\df\!\! &  R \ 
    \mid \ \pi_{\pmb{A}}(\query) \ 
    \mid \ \sigma_{\theta}(\query) \ 
    \mid \ \query \join \query' \
	\mid \ \query \cup \query' \ 
    \mid \ \query - \query' \\
    \theta  & \!\!\df\!\! & A =
   A' \mid \neg \theta \mid \theta \vee \theta \mid \theta \wedge \theta
\end{array}
$$ 
where $R$ ranges over relations in $\localsch$.
The sets of attributes of expressions $\attr{\query}$ are defined by extending $\attr{R}$, namely $\attr{\pi_{\pmb{A}}(\query)}$ is $\pmb{A}$, while both of $\attr{\sigma_{\theta}(\query)}$ and $\attr{ \query \op \query'}$ are $\attr{Q}$, for $\op$ being union and difference; $\attr{ \query \join \query'} = \attr{ \query} \cup \attr{  \query'}$ and $\attr{\rho_{ A \rightarrow A'}(\query)} = (\attr{\query} \setminus  \{ A \})\cup \{ A'\}$. 

The expressions of RA must satisfy the usual well-definedness rules: ${\pi_{\pmb{A}}(\query)}$ is well-defined if $\pmb{A}\subseteq \attr{\query}$; set operations are defined if $\attr{\query}=\attr{\query'}$, and for renaming from $A$ to $A'$ we must have $A\in\attr{\query}$ and 
$A'\not\in\attr{\query}$.
}

\OMIT{
\begin{align*}
	\query, \query' \ \ \df\ \ \  & R &\text{if $R \in \localsch$} \\  
	&\pi_{\pmb{A}} (\query) &\text{if $\pmb{A}\subseteq \attr{\query}$} \\ 
	&\sigma_{\theta}(\query) &
    \\
	&\rho_{ A \rightarrow  A'}(\query) &\text{if } A \in \attr{\query},A'\not \in \attr{\query}\\
	&\query \join \query' \\ 
	&\query \cup \query', \ \ \query \cap \query', \ \ \query - \query' &\text{if $\attr{\query} = \attr{\query'}$}  \\
\end{align*}
with $\theta, \theta' \df \ A = A' \mid \neg \theta \mid \theta \vee \theta' \mid \theta \wedge \theta'$ where $A,A'\in \attr{\query}$, and $\attr \query$ extending $\attr R$ by the following rules:

\begin{align*}
	\attr{\pi_{\pmb{A}} (\query)} \df\ \ \  & {\pmb{A}}& \\
	\attr{\sigma_{\theta} (\query)} \df\ \ \  & \attr{\query}& \\
	\attr{\rho_{ A \rightarrow A'}(\query)} \df\ \ \  & \left( \attr{\query} \setminus  \{ A \} \right)\cup \{ A'\}& \\
	\attr{ \query \op \query'} \df\ \ \  & \attr{ \query }& \text{ for }\op\in \{\cup,\cap,\setminus \}\\
	\attr{ \query \join \query'} \df\ \ \  & \attr{ \query} \cup \attr{  \query'}& \\
\end{align*}
}
The result of evaluation of a query $\query$ on a database $\db$ is a relation $\sem{\query}_{\db}$  over $\attr{\query}$ defined as:
\begin{align*}
	\sem{R}_{\db} \df\ \ \  & {\db}(R)&\\
	\sem{\pi_{\pmb{A}}(\query)}_{\db} \df\ \ \
	& \{ \tup \restriction_{\pmb{A}} \ \mid \ \tup \in \query \}&\\
	\sem{\sigma_{\theta}(\query)}_{\db} \df\ \ \  & \{ \tup \ \mid\ \tup \in \query \text{ and } \tup \models \theta\}&\\
	\sem{\rho_{ A \rightarrow A'}(\query)}_{\db}\df\ \ \  &\{ \rho_{ A \rightarrow  A'}(\tup) \mid \tup \in \query \}&\\
	\sem{\query \join \query'}_{\db} \df\ \ \
	& \{\mu \tupunion \mu' \ \mid \ \tup\in\query, \ \tup'\in \query' \}&\\
	\sem{\query \op \query'}_{\db} \df\ \ \  & \sem{\query}_{\db} \op \sem{ \query'}_{\db}& \text{ for }\op\in \{\cup,\setminus \}
\end{align*}

\noindent
with $\mu\models \theta$ having the standard semantics
$\mu \models A=A'$ iff $A,A'\in\dom{\mu}$ and $\mu(A) = \mu(A')$,
extended to Boolean connectives $\wedge, \vee, \neg$.

\paragraph{Core PGQ}\label{sec:corepgq}

\new{
Assume that for each variable $x\in\Vars$ and each key
$k\in\keyset$, both $x$ and $x.k$ belong to the set of attributes
$\mathcal A$. For each pattern $\pat$ and each output specification
$\return$, we have a relation symbol $R_{\pat,\return}$ whose set of
attributes are the elements of $\return$. 
Let $\globalschpat$ contain all such
relation symbols. 

\begin{definition}[Core PGQ]\label{core-pgq-def}
{\em Core PGQ} is defined as $\RA(\globalschpat)$, i.e.,
the set of relational algebra expressions over the schema
$\globalschpat$.
\end{definition}  

To define the semantics of Core PGQ queries, assume without loss of
generality that $\Valueset \subseteq \Univ$.  This ensures that
results of pattern matching are relations of the schema
$\globalschpat$, because for every path pattern with output
$\pat_{\return}$ and a property graph $\gdb$, the table 
$\sem{\pat_{\return}}_{\gdb}$ is an instance of relation
$R_{\pat,\return}$ from $\globalschpat$.
\OMIT{
\begin{proposition}\label{prop:multipatisrel}
	For every path pattern with output $\pat_{\return}$ and a property graph
	$\gdb$,  the set  $\sem{\pat_{\return}}_{\gdb}$ is an
	instance of  relation $R_{\pat,\return}$ from 
	$\globalschpat$.
\end{proposition}
}
%
Then the semantics of Core 
PGQ is simply the extension of the semantics of
$\RA$ defined above where for base relations we have
\OMIT{
extension of the semantics of $\RA$ 
as we only need to
define the semantics of base relations by
}
$ \sem{R_{\pat,\return}}_\gdb \ \df \ \sem{\pat_\return}_\gdb\,.$
\OMIT{
and then use the semantic rules for patterns from
Fig.~\ref{fig:sempathpattern} and for  $\RA$
 from 
Section \ref{sec:ra}.
}
}

%% file: sec-gql.tex
\new{We next provide a formal model of GQL. Recall that it shares
patterns with PGQ.
What is different is the way GQL processes results of
pattern matching: not in a bottom-up way with $\RA$ operators like
PGQ, but rather in sequential, or pipelined way where the output of
each operation in a sequence serves as the input for the next
operation. Using terminology adopted by Cypher \cite{cypher}, GQL
calls this {\em linear composition}. Unlike $\RA$, it lacks proper
formalization, and thus next we provide a formal description of a
different flavor of $\RA$, obtained by linear composition.
}

\OMIT{
as we showed in
Section 
In this section we formalize the notion of linear composition that
underlies Cypher and GQL, and is present independently of them in
purely relational languages such as PRQL. We use similar preliminaries
(schemas, databases) as in the previous section, then
present its linear-relational algebra variant, prove equivalence with relational algebra, and discuss
the origins of this approach in query language design. 
}

\subsection{LCRA: Linear Composition RA}\label{sec:lcra}
{
This language captures the sequential
(linear) application of relational operators as seen in Cypher, GQL,
and also PRQL. Its expressions over a schema $\localsch$, denoted by
$\LRA(\localsch)$,  
are defined as: }
%
\begin{align*}
    \textit{Linear Clause:   }\,\,\,  \linquery, \linquery' \df&~    S \mid \pi_{\pmb{A}} \mid \sigma_{\theta} \mid \rho_{ A\rightarrow A'} \mid \linquery\, \linquery' \mid \{\query\} &\\
   \textit{Query:   }\,\,\, \query, \query' \df&~ \linquery \mid \query\cap \query' \mid \query\cup \query' \mid\query \setminus \query' 
\end{align*}

where $S$ ranges over $\localsch$,  while $\pmb{A} \subseteq \mathcal A,$ and $A,A'\in\mathcal A$, and $\theta$ is defined as for RA. Unlike for RA, the output schema of LCRA clauses and queries can be determined only dynamically.

To define the semantics of a query $\query$ on a database $\db$ we
need a starting value of the driving table $\rel$, and this is taken
to be $I_{\emptyset}$, the relation containing only one empty tuple
$\mu_{\emptyset}$ where 
$\dom{\mu_{\emptyset}}\df \emptyset$, cf.~\cite{cypher,icdt23}.

The semantics $\sem{\,}_{\db}$ of LCRA clauses $\linquery$ and queries
$\query$ is a {\em mapping} from relations into relations (known as
{\em driving tables} for Cypher and GQL). It is defined as follows: 

\begin{align*}
    \sem{S}_{\db}(\rel) \df&~ \rel \join \db(S) \\
    \sem{\pi_{\pmb{A}}}_{\db}(\rel)\df&~ \{\tup\restriction_{\pmb{A}\cap \attr{\rel}} \  \mid \ \tup\in \rel\}\\
    \sem{\sigma_{\theta}}_{\db}(\rel) \df&~ \{\tup \  \mid \ \tup \in \rel,\, \tup \models \theta  \} \\
    \sem{\rho_{ A \rightarrow  A'}}_{\db}(\rel) \df&~ \{\rho_{ A \rightarrow  A'}(\tup) \mid \tup\in \rel,\, A'\not\in \domfunc{\tup} \}\\
     \sem{\linquery \linquery'}_{\db}(\rel) \df&~  \sem{\linquery'}_{\db}(\sem{\linquery}_{\db}(\rel))\\
     \sem{\{\query\}}_{\db}(\rel) \df&~ \rel \join \sem{\query}_{\db}(I_{\emptyset})  \\
    \sem{\query \op \query'}_{\db}(\rel) \df&~ \sem{\query}_{\db}(\rel) \op \sem{\query'}_{\db}(\rel),  \quad  \text{ for } \op \in\{\cup,\cap,-\} 
\end{align*}

\new{
A clause or a query of $\LRA$ always looks at a unique input relation
$\rel$. If a clause is a name of a relation $S$ or an entire query
$\{\query\}$, then it is joined with $\rel$. Projection, selection,
and renaming clauses behave in the usual way, and apply to the input
relation $\rel$. The meaning of set operations -- union, intersection,
difference -- is also standard, and two clauses $\linquery\linquery'$
are simply executed in sequence, with the result of $\linquery$ on
$\rel$ becoming the input to $\linquery'$. 

This semantics determines (dynamically) the set of attributes of
outputs of clauses and queries; for completeness we present it here:}
\begin{flalign*}
    \attr{\sem{S}_{\db}(\rel)} \df&~ \attr{S} \cup \attr{R} &\\
    \attr{ \sem{\pi_{\mathbf A}}_{\db}(\rel)}\df&~ \mathbf A \cap \attr{ \rel} &\\
    \attr{\sem{\sigma_{\theta}}_{\db}(\rel)} \df&~ \attr{\rel} &\\
    \attr{\sem{\rho_{ A \rightarrow A'}}_{\db}(\rel) } \df&~\attr{\rel} \setminus\{A'\} \cup \{ A\}&\\ 
    \attr{\sem{\{\query\}}_{\db}(\rel)} \df&~ \attr{\sem{\query}_{\db}(\rel)} \cup \attr{\rel}&\\
    \attr{\sem{\query \circ \query'}_{\db}(\rel)} \df&~ \attr{\sem{\query}_{\db}(\rel)}&
\end{flalign*}
\new{
Query output on a database is defined as
$\sem{\query}_{\db}( I_{\emptyset})$.
}

\subsection{Core GQL}

\new{
Just as we defined Core PGQ as $\RA$ over output of patterns, we now
define Core GQL as $\LRA$ over the same.
\begin{definition}[Core GQL]\label{core-gql-def}
  The language {\em Core GLQ} is defined as
the set of linear composition relational algebra expressions over 
the schema $\globalschpat$, i.e.,
$\LRA(\globalschpat)$.
\end{definition}  
As with Core PGQ, we define the semantics of Core GQL by
$ \sem{R_{\pat,\return}}_\gdb  \df  \sem{\pat_\return}_\gdb$
and then use the semantic of $\LRA$ above. 

We now present an example to explain how GQL's linear composition
works.
}
We use a simplified query based on the money laundering query
from the introduction. It looks 
for someone who has two friends in a city different from theirs, and
outputs the person's name and account (\texttt{Porthos} and \texttt{a2} in our example): 

\begin{gql}
	MATCH (x)-[:Friends]->(y)-[:Friends]->(z),
	(y)-[:Owns]->(acc_y)
	FILTER (y.city) <> (x.city) 
    AND (x.city=z.city)
	RETURN y.name AS name, acc_y AS account
\end{gql}
\noindent
The equivalent Core GQL formula is 

\smallskip
$	R_{\pat_1,\return_1}  
	\quad R_{\pat_2, \return_2} \
	\quad \sigma_{y.\text{city}\neq x.\text{city} \wedge x.\text{city}=z.\text{city}} 
	\quad\pi_{y.\text{name},\text{acc}\_y}\\ 
	{} \quad\quad\quad\quad\rho_{y.\text{name} \to \text{name}}  
	\quad\rho_{acc\_y \to \text{account}}
$
\OMIT{
$\begin{array}{l}
	R_{\pat_1,\return_1} \\ 
	\quad R_{\pat_2, \return_2} \\
	\quad\quad \sigma_{y.\text{city}\neq x.\text{city} \wedge x.\text{city}=z.\text{city}} \\ 
	\quad\quad\quad\pi_{y.\text{name},\text{acc}\_y}\\ 
	\quad\quad\quad\quad\rho_{y.\text{name} \to \text{name}} \\ 
	\quad\quad\quad\quad\quad\rho_{acc\_y \to \text{account}}
\end{array}
$
}

\smallskip
\noindent
where 
$$
\begin{array}{rcl}
	\pat_1 & \df & \big((x) \overset{e_1}{\rightarrow} (y) \ (y)
	\overset{e_2}{\rightarrow}(z)\big)\langle \text{Friends}(e_1) \wedge
	\text{Friends}(e_2)\rangle
	\\
	\pat_2 & \df & \big((y) \overset{e_3}{\rightarrow} (\text{acc}\_y)\big)\langle
	\text{Owns}(e_3)\rangle
	\\
	\return_1 & \df & (x,\ y,\ z,\ x.\text{city},\ y.\text{city},\ z.\text{city})
	\\
	\return_2 & \df & (y, \ \text{acc}\_y)\,.
\end{array}
$$

\subsection{Equivalence of Core PGQ and Core GQL}
\OMIT{We denote by $\RA(\globalsch)$ the set of all queries in the infinite union
\[
\bigcup_{\localsch\subset \globalsch \text{ is a local schema}} \RA(\localsch)
\]
and by 
$\LRA(\globalsch)$, 
\[
\bigcup_{\localsch\subset \globalsch \text{ is a local schema}} \LRA(\localsch)
\]
}
{
We say that two queries $\query_1, \query_2$ (possibly from different
languages) are \emph{equivalent} if $\sem{\query_1}_{\db} =
\sem{\query_2}_{\db}$
for every database $\db$. 
A query language $L_1$ is \emph{subsumed by}  $L_2$ if for each query $\query_1$ 
in $L_1$ there is an equivalent query $\query_2$ 
in $L_2$.
If there is also a query $\query_2 \in L_2$ for which there is no equivalent query $\query_1\in L_1$ 
then $L_1$ is said to be \emph{strictly less expressive than } $L_2$.
Finally
$L_1$ and $L_2$ are {\em equivalent} if $L_1$ is {subsumed by}  $L_2$ and 
$L_2$ is {subsumed by}  $L_1$.

\begin{theorem}\label{thm:ravslra}
Languages  $\RA(\localsch)$ and
$\LRA(\localsch)$ are equivalent, for every schema $\localsch$
\end{theorem}

Thus, $\LRA$ proposed as the relational processing engine of graph
languages like Cypher and GQL is the good old $\RA$ in a slight
disguise. 
As an immediate consequence of Theorem~\ref{thm:ravslra} we have:
\begin{corollary}\label{cor:pgqvsgql}
	The languages Core PGQ and Core GQL have the same expressive
	power. 
\end{corollary}

Notice that the definition of linear clauses and queries of $\LRA$ are
mutually recursive as we can feed any query $\query$ back into clauses
via $\{\query\}$ (this corresponds to the \sqlkw{CALL} feature of Cypher
and GQL). If this option is removed, and linear clauses are not
dependent on queries, we get a simplified language $\sLRA$ (simple
$\LRA$).
Specifically, in  this language linear clauses are given by the
grammar
$\linquery \df  S \mid \pi_{\mathbf A} \mid \sigma_{\theta} \mid \rho_{ A\rightarrow
A'} \mid \linquery \linquery$.
To see why the $\{\query\}$ clause was necessary in the definition of $\LRA$,
we show

\begin{proposition}\label{prop:SLRA-vs-LRA}
    $\sLRA$ is strictly less expressive than $\LRA$. 
\end{proposition}

\OMIT{ 
\noindent
{\em Proof}. At first, observe that simple linear clauses (without
$\{\query\}$) can only express conjunctive queries, as their semantics
applies operations $\join, \pi, \sigma$ and renaming to base
relations. Hence, queries of sLCRA are Boolean combinations of
conjunctive queries, known as BCCQs. While it appears to be folklore
that BCCQs are strictly contained in first-order logic, we were unable
to find the simple proof explicitly stated in the literature, hence we offer ohe
here.

Consider a vocabulary of a single unary predicate $U$ and databases
$\db_1$ and $\db_2$ such that  $\db_1(U) = \{a_1\}$ and
$\db_2(U) = \{a_1,a_2\}$ for two different constants $a_1$ and $a_2$. Since
$\db_1$ and $\db_2$ are homomorphically equivalent, they agree on all
conjunctive queries, and therefore on all BCCQs, but they do not agree
on first-order (and hence $\RA$) query that checks if relation $U$ has
exactly one element. 
\qed
}

\subsection{The origins of linear composition}

\new{Linear composition features prominently in graph languages
(Cypher, GQL) and also some relational languages (PRQL); at the same
time it had not been formalized nor studied as its non-linear
relational algebra analog (until now). Thus we use this short section to briefly
explain the origins of linear composition.
}
Linear composition was introduced in the design of
Cypher \cite{cypher} as a way to bypass
the lack of a compositional language for
graphs. Specifically, pattern matching
transforms graphs into relational tables, and other Cypher operations
modify these tables. If we have two such read-only queries
$\gdb \to \mathbf{T}_1$ and $\gdb \to \mathbf{T}_2$ from graphs to
relational tables, it is not
clear how to compose them. To achieve composition, Cypher read-only
queries are of the form $\query: \gdb \times \mathbf{T} \to \mathbf{T}'$ (and its
read/write queries are of the form
$\query: \gdb \times \mathbf{T} \to \gdb' \times \mathbf{T}'$). In
other words, they turn a graph {\em and a table} into a table. Thus, the
composition of two queries
$\query_1, \query_2: \gdb \times \mathbf{T} \to \mathbf{T}'$ is their {\em linear
composition} $\query_1 \ \query_2$ which on a graph $\gdb$ and table
$\mathbf{T}$ returns $\query_2\big(\gdb,\query_1(\gdb,\mathbf{T})\big)$. 

Independently, the same approach was adopted by a relational language PRQL \cite{prql},
where P stands for ``pipelined'' but the design philosophy is identical. For example one could write 

\medskip
\noindent 
{\small\tt 
\sqlkw{FROM} {\tt R} \sqlkw{FILTER} A=1 \sqlkw{JOIN:INNER} S \sqlkw{FILTER} B=2 \sqlkw{SELECT} C, D}

\medskip
\noindent
with each clause applied to the output of the previous clause. The
above query is the same as relational algebra query

\[\pi_{C,D}\Big(\sigma_{B=2}\big(\sigma_{A=1}(R) \Join S\big)\Big)\]

Though PRQL design is relations-to-relations, the motivation for
pipelined or linear composition comes from creating a database analog
of {\tt dplyr} \cite{dplyr}, a data manipulation library in R, that can be translated 
to SQL. While {\tt dplyr}'s operations are very much relational in
spirit, it is integrated into a procedural language, and hence the
imperative style of programming was inherited by PRQL and also adopted by the piped syntax of SQL \cite{google-pipes}.



\smallskip
\new{
To recap, we defined simple
theoretical abstractions
{\em Core PGQ} of SQL/PGQ and
{\em Core GQL} of GQL, that
share the same {pattern matching language} 
    turning graphs into relations;
then on top of pattern matching outputs we use $\RA$ for PGQ and
$\LRA$ for GQL.
}
It should be kept in mind that these abstractions capture
the essense of SQL/PGQ and GQL in the same way as $\RA$
and first-order logic capture the essence of SQL: they define a
theoretical core that is amenable to a formal study, but real
languages, be it 500 
pages of the GQL standard or over 4000 pages of the SQL standard,
have many more features.

%% file: sec-express.tex
\new{
A common pattern of query language development is this: first, a
small core language is designed; its focus is on
declarativeness and optimizability. As a consequence the
expressiveness of such a language is limited. 
These limitations are typically observed in practice by programmers'
inability to write certain queries and often later proven formally,
for a theoretical language that defines the key features of
the practical one. Using these inputs as motivations, extra features are added to the
real-life language if practical applications warrant this and
user demands.
}

A classical example of this development cycle is SQL and recursion. 
In this case, we had a ``canonical'' query whose
inexpressibility it was important to show in order to justify
additions to the language. The query was transitive closure of
relation (or, equivalently, testing for graph
connectivity \cite{GaifmanVardi85,Fagin75,AhoUllman79}).
Thus, before embarking on the study of the expressiveness of GQL and
SQL/PGQ, we ask for an analog of such ``canonical'' queries for them
that users want to express but appear to be unable to. 

Fortunately, we have such queries, thanks to the discussions already happening in
the ISO committee that maintains both SQL and GQL
standards \cite{tobias,fred}. \new{In fact much of recent effort tries to
repair what appears to be a hole in the expressiveness of the language. 
Specifically, it seems to be impossible to express queries that impose
conditions on how values of edges properties change along the paths,
even if the same conditions can be expressed for properties of nodes.
To give a simple example of such a condition, consider the following:}
\begin{itemize}
\item we {\em can} check, by a very simple pattern,  if there is a path of
transfers between two accounts where balance in intermediary accounts
(values held in {\em nodes})
increases, but
\item it appears that we {\em cannot} check if there is a path of
transfers between two accounts where timestamp
(values held in {\em edges})
increases.
\end{itemize}
This perceived deficiency has already led to early proposals to significantly
enhance pattern matching capabilities of the language \cite{tobias,fred},  in a
way whose complexity implications appear significant but are not fully understood.

Before such a dramatic expansion of the language gets a stamp of
approval, it would be nice to know whether it is actually needed.
The goal of this section is to provide both theoretical and
experimental evidence that we do need extra language features to
express queries such as {\em ``increasing values in edges''}. 

\OMIT{
To analyze the expressiveness of patterns, there is a natural
way of defining queries given by
them, in the spirit of RPQs which
define pairs of nodes connected by a path. 
Take any pattern $\pat$ and turn it into a pattern
with output
\begin{equation}
\label{pat-to-rpq-eq}
\big((x_s)\ \pat \ (x_t)\big)_{x_s,x_t}
\end{equation} that will output pairs
of source and target nodes $x_s$ and $x_t$ connected by a path
satisfying $\pat$.
If we are given a query that maps a property graph $\gdb$ (possibly
from a class of graphs $\mathcal C$) to a pair
of nodes, we say it is expressible by a pattern $\pat$ if its output
is given by pattern with output (\ref{pat-to-rpq-eq}) on every graph
(from $\mathcal C$).

We now start with these results, before looking at conditions on
paths as a whole. We conclude this section by adjusting our definition
to model Cypher patterns as they originally appeared \cite{cypher} and
confirm a previously unproven folklore result that Cypher patterns
cannot express all RPQs.
}

\subsection{What \emph{can} be expressed}
\label{can-be-expressed-sec}
A pattern $\pat$ can be viewed as a  query $\big((x_s)\ \pat \
(x_t)\big)_{x_s,x_t}$ that returns endpoints (source and target) of the paths matched by
$\pat$ as values of attributes $x_s$ and $x_t$.
\OMIT{
analyze the expressiveness of patterns, there is a natural
way of defining queries given by
them, in the spirit of RPQs which
define pairs of nodes connected by a path. 
Take any pattern $\pat$ and turn it into a pattern
with output
\begin{equation}
\label{pat-to-rpq-eq}
\big((x_s)\ \pat \ (x_t)\big)_{x_s,x_t}
\end{equation} that will output pairs
of source and target nodes $x_s$ and $x_t$ connected by a path
satisfying $\pat$.
}
Earlier queries are defined formally as:

\begin{itemize}
\item $\qNodeLessThan$ returns endpoints of a path along which the value
of property $k$ of {\em nodes} increases.  It returns pairs $(u_0,u_n)$
of nodes such that there is path $\pathval(u_0,e_1,u_1,\ldots,e_n,u_n)$ so that $u_0.k < u_1.k < \cdots < u_n.k$. 
\item $\qEdgeLessThan$ returns endpoints of a path along which the value
of property $k$ of {\em edges} increases. It returns pairs $(u_0,u_n)$
of nodes such that there is path $\pathval(u_0,e_1,u_1,\ldots,e_n,u_n)$ so that $e_1.k < e_2.k < \cdots < e_n.k$.
\end{itemize}

It is a simple observation that $\qNodeLessThan$ is expressible by
the Core PGQ and GQL pattern with output
$$
(x_s) \ \Big(\, \big((x)\rightarrow(y)\langle
x.k<y.k \rangle \big)^{0..\infty}\,\Big)\ (x_t).\,$$
\new{
In fact this can be generalized to {\em order motifs} which are
defined as strings over the alphabet $\{\ua,\da\}$. Given such a
string $w$, the query $\qNode_w$ matches paths which can be decomposed
according to $w$ so that in each segment corresponding to $\ua$ values
of property $k$ increase and in each segment corresponding to $\da$
these values decrease. For example, $\qNode_{\ua\da\ua}$ matches
arbitrary length paths
where values of property $k$ of nodes first increases, then decreases,
and then increases again. With the same approach as the above query,
we can see that  $\qNode_w$ is expressible for every order motif $w$. 
}

What happens when we move from nodes to edges? The same approach
will not work for  $\qEdgeLessThan$: if 
 we (erroneously) express it  by 
 $$
\big(\,()\overset{x}{\rightarrow}()\overset{y}{\rightarrow}()\langle x.k<y.k \rangle\, \big)^{0..\infty}
$$
it fails: on the input
$()\overset{3}{\rightarrow}()\overset{4}{\rightarrow}()\overset{1}{\rightarrow}()\overset{2}{\rightarrow}()$
(where the numbers on the edges are values of property $k$) it
returns the start and the end node of the path, even though values in
edges do not increase. This is due to the semantics of path
concatenation: two paths concatenate if the last node of the first path
equals the first node of the second, and therefore conditions on
edges in two concatenated or repeated paths are completely ``local''
to those paths.

In some cases, where graphs are of special shape, we can obtain the
desired patterns by resorting to tricks that a normal query language user
would be rather unlikely to find. Specifically, consider property graphs
whose underlying graph structures are just paths. That is, we look at 
\emph{annotated paths} $P_n$, which are of the form 

\begin{center}
    \begin{tikzpicture}
    \node at (0,0) [shape=circle, draw] (node1) {$v_0$};
    \node at (1.5,0) [shape=circle] (invisiNode) {\ldots};
    \node at (3,0) [shape=circle, draw] (node2) {$v_n$};
    \draw[->] (node1.east) -- (invisiNode.west) node [above,align=center,midway]{$e_0$};
    \draw[->] (invisiNode.east) -- (node2.west) node [above,align=center,midway]{$e_{n-1}$};
\end{tikzpicture}
\end{center}

\noindent where $v_0,\ldots,v_n$ are distinct nodes, $e_0, \ldots,
    e_{n-1}$ are distinct edges (for $n > 0$), and each edge $e_i$
    has property $e_i.k$.
    Let
$\pat^{\Edgeset}_{<}$ be 
$$\begin{array}{cl}
& (x_s) \left(\big(
(u)\overset{x}{\rightarrow}(z)  \overset{y}{\rightarrow}(v) \overset{w}{\leftarrow}
(z)\big) \langle x.k < y.k \rangle\right)^{1..\infty}
\to (x_t)\\ + & (x_s)\to(x_t)
\end{array}
$$
\new{A pattern with output $\pat_\return$ expresses a query $Q$ on $G$ if 
 $\sem{\pat_\return}_G = Q(G)$.
}
\begin{proposition}\label{prop:e<express}
The pattern 
$\big(\pat^{\Edgeset}_{<}\big)_{(x_s,x_t)}$ 
expresses $\qEdgeLessThan$ on annotated paths.
\end{proposition}
In the first disjunct, the forward edge $\overset{y}{\rightarrow}$ serves as a
look-ahead that enables checking whether the condition holds, and the
backward edge $\overset{w}{\leftarrow}$ enables us to continue constructing the path in the next
iteration from the correct position $(z)$. The fact that the query
operates on $P_n$ ensures that the last edge to $x_t$ does not violate the
condition of the query as it was already traversed in the iterated subpattern. The second disjunct takes care of
paths of length $1$. 

\subsection{What \emph{cannot} be expressed with patterns}

\new{To achieve expressibility of $\qEdgeLessThan$ in Proposition~\ref{prop:e<express} we made two strong assumptions: not only is the input of a  very special shape (a directed path with forward edges), but also the pattern uses backwards edges; the latter would be natural for an  
{\em oriented} path (in which edges can go in either
direction \cite{HN-book}), but looks rather unnatural in this setting. Thus, we ask ourselves whether $\qEdgeLessThan$ can be expressed in a {\em natural}
way. In fact we can pose a more general question about queries
$\qEdge_w$ for arbitrary order motifs; in Section \ref{can-be-expressed-sec} we saw that on nodes, all
order motifs can be expressed.
}

To formalize what we mean by ``natural'', define
\emph{one-way path patterns} by a restriction of the
grammar:
\begin{align*}
\pat \ \df \ & 
(x) \ \mid\ \overset{x}{\rightarrow} \ \mid \  \pat_1
+ \pat_2 \ \mid \  \pat^{n..m} \ \mid \ \pat_{\langle \theta \rangle} \ \mid\   \pat_1\, \pat_2
\end{align*}
where we require that $\sch{ \pat_1}\cap \sch{\pat_2} = \emptyset$ in
$\pat_1\, \pat_2$ (and variables $x$, as 
before, are  optional). 
The omission of backward edges $\overset{x}{\leftarrow}$ in one-way
patterns is quite intuitive. The restriction on
variable sharing in concatenated patterns is because 
backward edges can be simulated by simply  repeating variables, as 
done above in $\pat^{\Edgeset}_{<}$.


If there is a natural way, easily found by programmers, of writing the
$\qEdgeLessThan$ query in PGQ 
and GQL, one would expect it to be done without backward
edges. Yet, this is not the case. \new{In fact, no order motif on
edges can be captured by GQL and PGQ patterns.  
\begin{theorem}\label{thm:Q<edge}
    There is no order motif $w$ such that $\qEdge_w$ is expressible by
    a one-way path pattern query. 
\end{theorem}

The inexpressibility of $\qEdgeLessThan$ is then an immediate
corollary for $w=\ua$.
}
\new{This result is a direct consequence of a general pumping argument
applied to annotated paths accepted by one-way patterns. 
\begin{theorem}
    For every one-way path pattern $\pat$, if for every $n\in \mathbb{N}$ there exists an annotated path $p$ of length $n$ accepted by $\psi$ then there exists $n_0\in \mathbb{N}$ such that for every annotated path $p$ accepted by $\pat$ with $|p|>n_0$, the path $p$ can be decomposed as  $p=p_1 p_2 p_3$, $|p_2|>1$ and for every $n\in \mathbb{N}$, the annotated path $p_1 p_2^{n+1}p_3 $ is also accepted by $\pat$.
\end{theorem}
Note that in the newly constructed path with repetitions, we assign new ids to the different occurrences of the elements in $p_2$, while keeping the data (annotations) unchanged.
The idea behind the proof is as follows: Since, by definition, there are infinitely many annotated paths that conform to 
$\pat$, it must exhibit unbounded repetition. However, because the semantics disregard variables that occur within unbounded repetitions, the transfer of information between iterations is limited. This restriction allows us to repeat parts of the annotated path while preserving the same semantics and, consequently, still conforming to
$\pat$.

We can derive a further corollary of this theorem that reinforces the intuition that GQL and PGQ patterns are incapable of capturing data values along a path ``as a whole.''

A canonical example of such a condition is checking whether all values
of a specific property of nodes or edges along a path are distinct. 
Formally, we define:  
\begin{itemize}
\item $\qdisjnodes$
returns pairs $(u_0,u_n)$
of nodes such that there is path
$\pathval(u_0,e_1,u_1,\ldots,e_n,u_n)$ so that $u_i.k \neq u_j.k$ for
all $0 \leq i < j \leq n$.
\end{itemize}
and $\qdisjedges$ is defined likewise but for edges in place of nodes.
One-way path
patterns do not have enough power to express such queries.
\begin{corollary}\label{thm:QNdiff}
   No one-way path pattern expresses $\qdisjnodes$ nor $\qdisjedges$. 
\end{corollary}
}

\new{\subsection{What \emph{cannot} be expressed in full GQL}}
Having examined the expressiveness of patterns, we now look at the
entire query languages Core GQL and Core PGQ and
compare them with recursive SQL and linear Datalog.
Recursive SQL is a good comparison target as the
 most natural relational language into which graph queries involving
 pattern with arbitrary length paths are
 translated \cite{yakovets,tassos}. 
The theoretical basis for recursive SQL common table expressions is
 {\em linear Datalog}, i.e., the fragment of Datalog in which definitions can refer to
recursively defined predicates at most once. This is precisely the
 restriction of recursive SQL: a recursively defined table can
 appear at  most once in \sqlkw{FROM}. 

Of course recursive SQL can be enormously powerful: as already
mentioned, combining recursion and arithmetic/aggregates one can
simulate Turing machines \cite{AbiteboulV95}. Thus, to make the comparison
fair, we look at {\em positive recursive SQL}: this is the fragment of recursive SQL
where subqueries can only define equi-joins. In other words, they only
use conjunctions of equalities in \sqlkw{WHERE}. Such a language just
adds recursion on top of unions of conjunctive queries, and ensures
termination and tractable data complexity \cite{ABLMP21}.   
We show that even these simple fragments of SQL and Datalog can
express queries that Core GQL and PGQ cannot define. 


\begin{theorem}\label{main-gql-thm}
There are queries that are expressible in positive recursive SQL, and in
linear Datalog, and yet are not expressible
in Core GQL nor Core PGQ.
\end{theorem}

At the end of the section, we explain why this is quite
surprising in view of what we know about Core GQL and PGQ:
complexity-theoretic considerations strongly suggest these should
define {\em all} queries from linear Datalog, and yet this is not the
case due to subtle deficiencies in the language design, which we
outline in Section \ref{sec:concl}.

To talk about expressing 
property graph queries in relational languages, 
we must represent graphs as relations. There
are many possibilities, and it does not matter (for showing
Theorem \ref{main-gql-thm}) which
one we choose, as these different representations are inter-definable
by means of unions of select-project-join queries. 
Since graphs in the separating example have labels and do not have property
values, we use an encoding consisting of unary
relations $N_\ell$ storing  $\ell$-labeled nodes, and binary relations $E_\ell$ storing pairs of nodes
$(n_1,n_2)$ with an $\ell$-labeled edge between
them.

To sketch the idea of the separating query, 
we define {\em dataless paths} as graphs $\gdb_n$, $n > 0$, with
nodes $v_0,\ldots,v_n$,
 edges $(v_0,v_1), \ (v_1,v_2), \ldots,
    (v_{n-1},v_n)$, where  
      $v_0$ has label  $\mathit{min\_elt}$ and $v_n$ has label $\mathit{max\_elt}$.
The separating query asks: is $n$ a power of 2?

The inexpressibility proof of this is based on showing that GQL
queries can only define 
Presburger properties of lengths of dataless paths, by
translating Core GQL queries on such paths into formulae of Presburger
Arithmetic (i.e., the first-order theory of $\langle \mathbb{N}, +,
<\rangle$).
Therefore the only definable properties of 
lengths are semilinear sets \cite{ginsburg-spanier}. Semi-linear
subsets of $\mathbb{N}$ are known to be ultimately periodic: that is,
for such a set $S\subseteq \mathbb{N}$, there exists a threshold $t$
and a period $p$ such that $n\in S$ if and only if $n+p\in S$, as long
as $n>t$.
Clearly the set $\{2^k \mid
k \in \mathbb{N}\}$ is not such. 

\OMIT{
\begin{sql}[float,caption=Power of 2 in SQL,captionpos=b,label=sql-power-2-query]
WITH RECURSIVE ADD(A, B, C) AS 
  ( (SELECT P.A, MIN_ELT.A AS B, P.A AS C
     FROM P, MIN_ELT) 
   UNION
    (SELECT MIN_ELT.A, P.A AS B, P.A AS C
     FROM P, MIN_ELT) 
   UNION
    (SELECT ADD.A, P1.B, P2.B AS C
     FROM ADD, P P1, P P2
     WHERE ADD.B=P1.A AND ADD.C=P2.A) ),
  POW2(A, B) AS 
  ( (SELECT P2.A, P2.B 
     FROM MIN_ELT M, P P1, P P2 
     WHERE M.A=P1.A AND P1.B=P2.A) 
  UNION                      
    (SELECT P.B AS A, ADD.C AS B
     FROM POW2, ADD, P
     WHERE POW2.A=P.A AND ADD.A=POW2.B
           AND ADD.B=POW2.B) )
(SELECT 'YES' FROM POW2, MAX_ELT 
 WHERE POW2.B=MAX_ELT.A)
\end{sql}
}

The query can be expressed in positive recursive SQL: 
\begin{sql}
WITH RECURSIVE ADD(A, B, C) AS 
  ( (SELECT P.A, MIN_ELT.A AS B, P.A AS C
     FROM P, MIN_ELT) 
   UNION
    (SELECT MIN_ELT.A, P.A AS B, P.A AS C
     FROM P, MIN_ELT) 
   UNION
    (SELECT ADD.A, P1.B, P2.B AS C
     FROM ADD, P P1, P P2
     WHERE ADD.B=P1.A AND ADD.C=P2.A) ),
  POW2(A, B) AS 
  ( (SELECT P2.A, P2.B 
     FROM MIN_ELT M, P P1, P P2 
     WHERE M.A=P1.A AND P1.B=P2.A) 
  UNION                      
    (SELECT P.B AS A, ADD.C AS B
     FROM POW2, ADD, P
     WHERE POW2.A=P.A AND ADD.A=POW2.B
           AND ADD.B=POW2.B) )
(SELECT 'YES' FROM POW2, MAX_ELT 
 WHERE POW2.B=MAX_ELT.A)
\end{sql}

\noindent 
The path is given by a binary relation
{\tt P} containing pairs $(v_i,v_{i+1})$ while 
minimal/maximal elements $v_0$ and $v_n$ are given by unary relations {\tt MIN\_ELT} and
{\tt MAX\_ELT}. The first common table expression
defines a ternary relation  {\tt ADD}
with  tuples $(v_i,v_j,v_k)$ such that
$i+j=k$. If $v_i$ is in {\tt MIN\_ELT}, then $(v_i,v_j,v_j)$
and $(v_j,v_i,v_j)$ are
in {\tt ADD} for all $j$ (the basis of recursion), and if
$(v_i,v_j,v_k)$ is in {\tt ADD} then so is $(v_{i+1},v_j,v_{k+1})$
(the recursive step). 
After that {\tt POW2} builds a relation with tuples
$(v_i,v_j)$ for $j = 2^i$. Indeed, if $(v_i,v_j)$
is in {\tt POW2}, then so is $(v_{i+1},v_k)$ for  $k=2\cdot j$,
which is tested by $(v_j,v_j,v_k)\in\texttt{ADD}$. The length of the
path is a power of 2 if
the second projection of
{\tt POW2} contains {\tt MAX\_ELT}.

Note that {\tt ADD} is defined by a linear
Datalog program and {\tt POW2} is defined by a linear Datalog program
that uses {\tt ADD} as EDB. Hence, the entire query is
defined by a piece-wise linear program (where already
defined predicates can be used as if they were EDBs). It is known that
such programs can be expressed in linear Datalog \cite{Afrati-tcs03}.


\paragraph{Complexity-theoretic considerations.}
We now explain why the inexpressibility result is quite
unexpected.  Note that all graph pattern languages can express the
reachability query. It is  complete for the complexity class 
\nlog\
via first-order reductions \cite{imm-book}:  an extension of 
first-order logic with the
reachability predicate capturtes precisely all \nlog\ queries.
Since GQL and PGQ can
emulate relational algebra (and thus first-order logic) over
pattern matching results, it appears that they should be able to express all
\nlog\ queries.
In fact graph query
languages expressing all \nlog\ queries have been known 
for a long time, starting with 
GraphLog \cite{crpq}, which introduced the ubiquitous notion
of CRPQs. 

\OMIT{
Graph query languages with data complexity  in \nlog\ have been
known for a long time. The  early
language GraphLog \cite{crpq}, which introduced the ubiquitous notion
of CRPQs,  was shown to capture the transitive closure logic and
thus \nlog\ (over ordered graphs). 
GraphLog's approach is similar to Datalog: CRPQs define new relations, which can then be used to define further CRPQs, much like how idb relations define new idb relations in Datalog.
}

However, Theorem~\ref{main-gql-thm} not only refutes the complexity-based intuition,
but in fact the separating query has an even lower \dlog\
complexity. Indeed, one can traverse the dataless path graph
while maintaining the counter
that needs a logarithmic number of bits.
This limitation highlights a deficiency of the language design:
despite having access to reachability and full power of first-order
logic on top of it, Core GQL and PGQ fall short of a declarative
language that is first-order logic with the reachability
predicate. The reason for this deficiency, 
that should ideally be addressed in future revisions of the standards,
is discussed below in  Section \ref{sec:concl}.

%% file: sec-experiments.tex
\begin{figure*}
	\centering
	\begin{subfigure}{0.3\textwidth}
		\includegraphics[width=\textwidth]{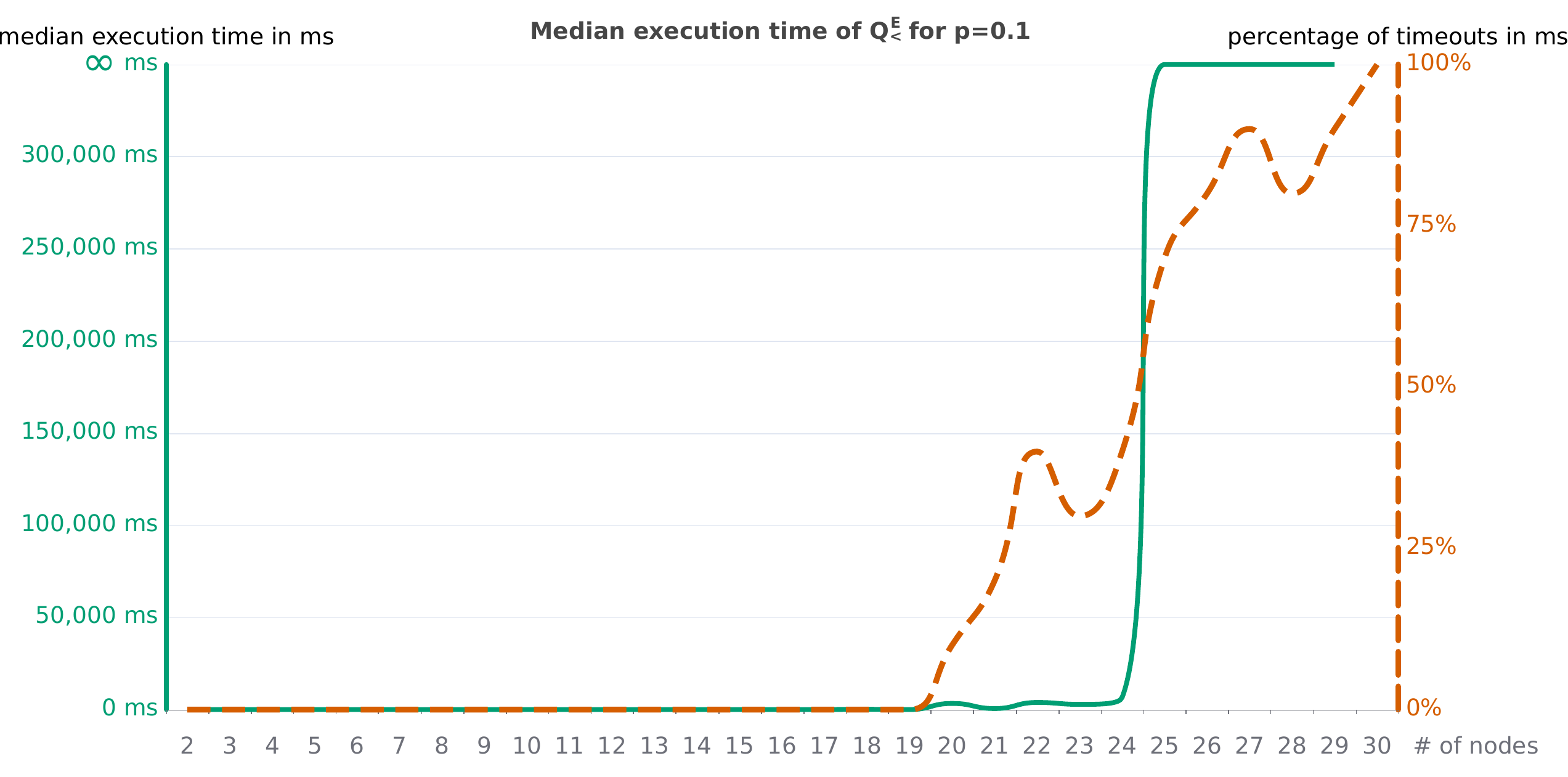}
		\caption{$p=0.1$}
		\label{neo4j-test-0.1}
	\end{subfigure}
	\begin{subfigure}{0.3\textwidth}
		\includegraphics[width=\textwidth]{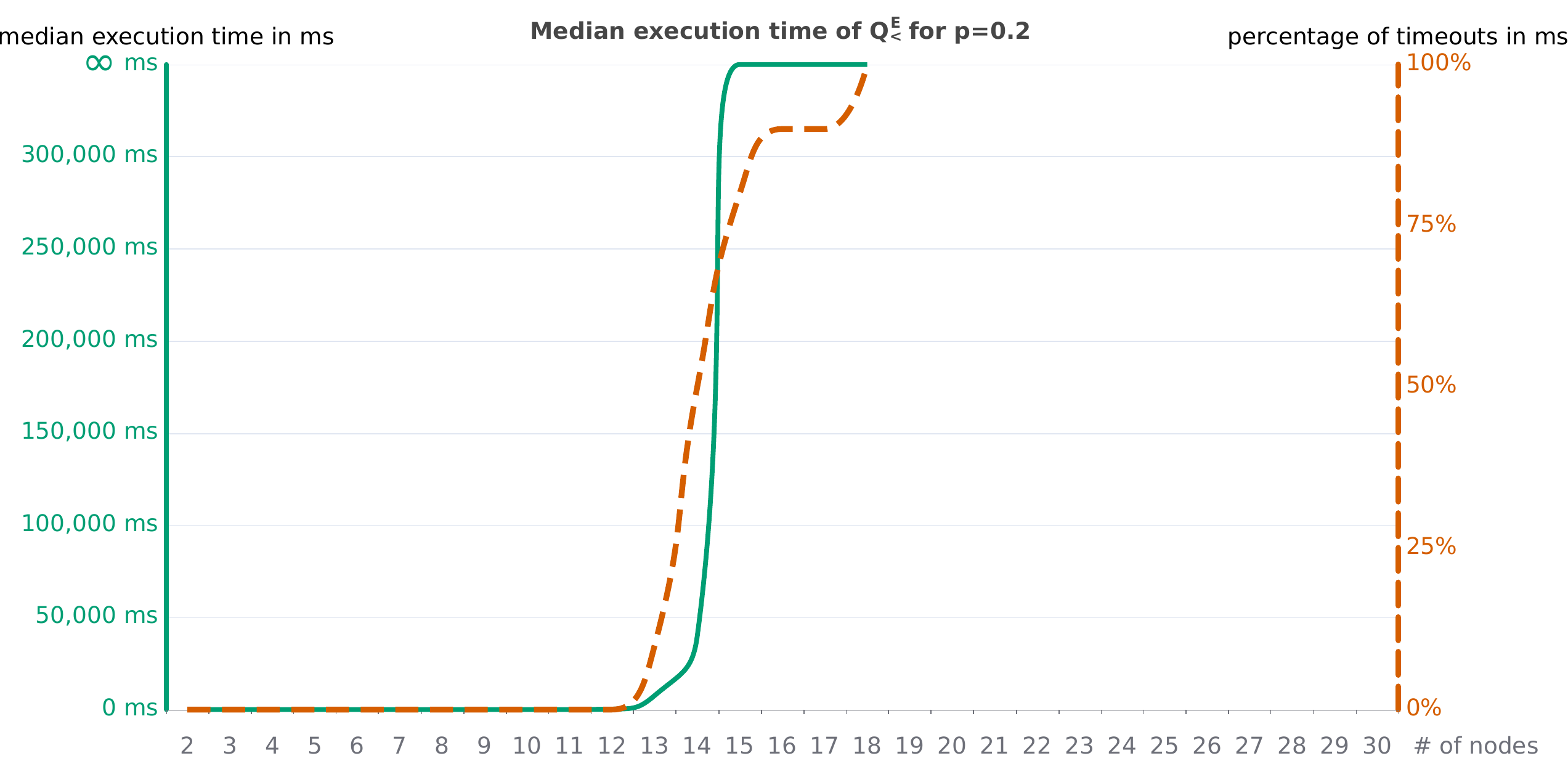}
		\caption{$p=0.2$}
		\label{neo4j-test-0.2}
	\end{subfigure}
	\begin{subfigure}{0.3\textwidth}
		\includegraphics[width=\textwidth]{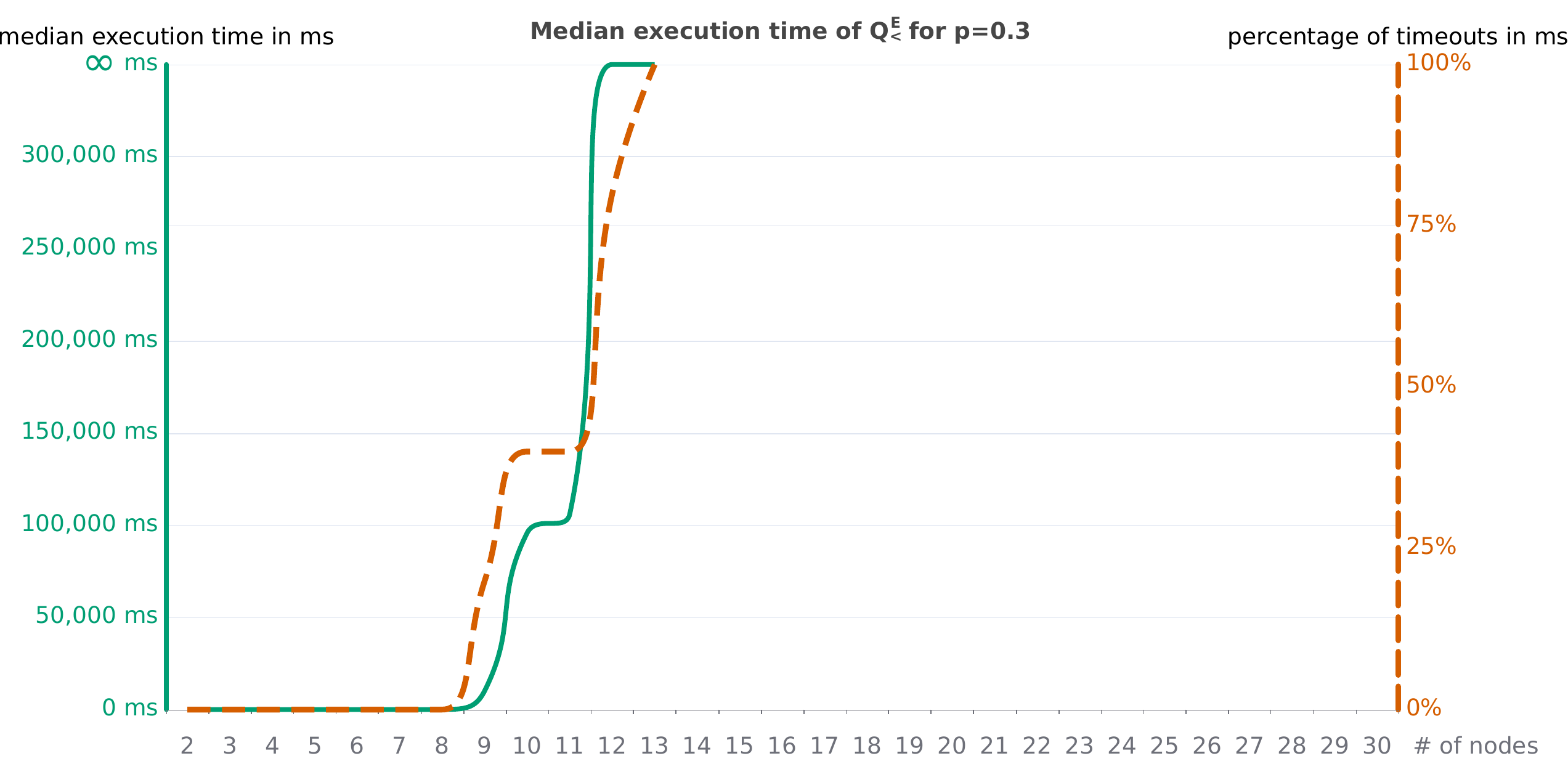}
		\caption{$p=0.3$}
		\label{neo4j-test-0.3}
	\end{subfigure}
	\begin{subfigure}{0.3\textwidth}
		\includegraphics[width=\textwidth]{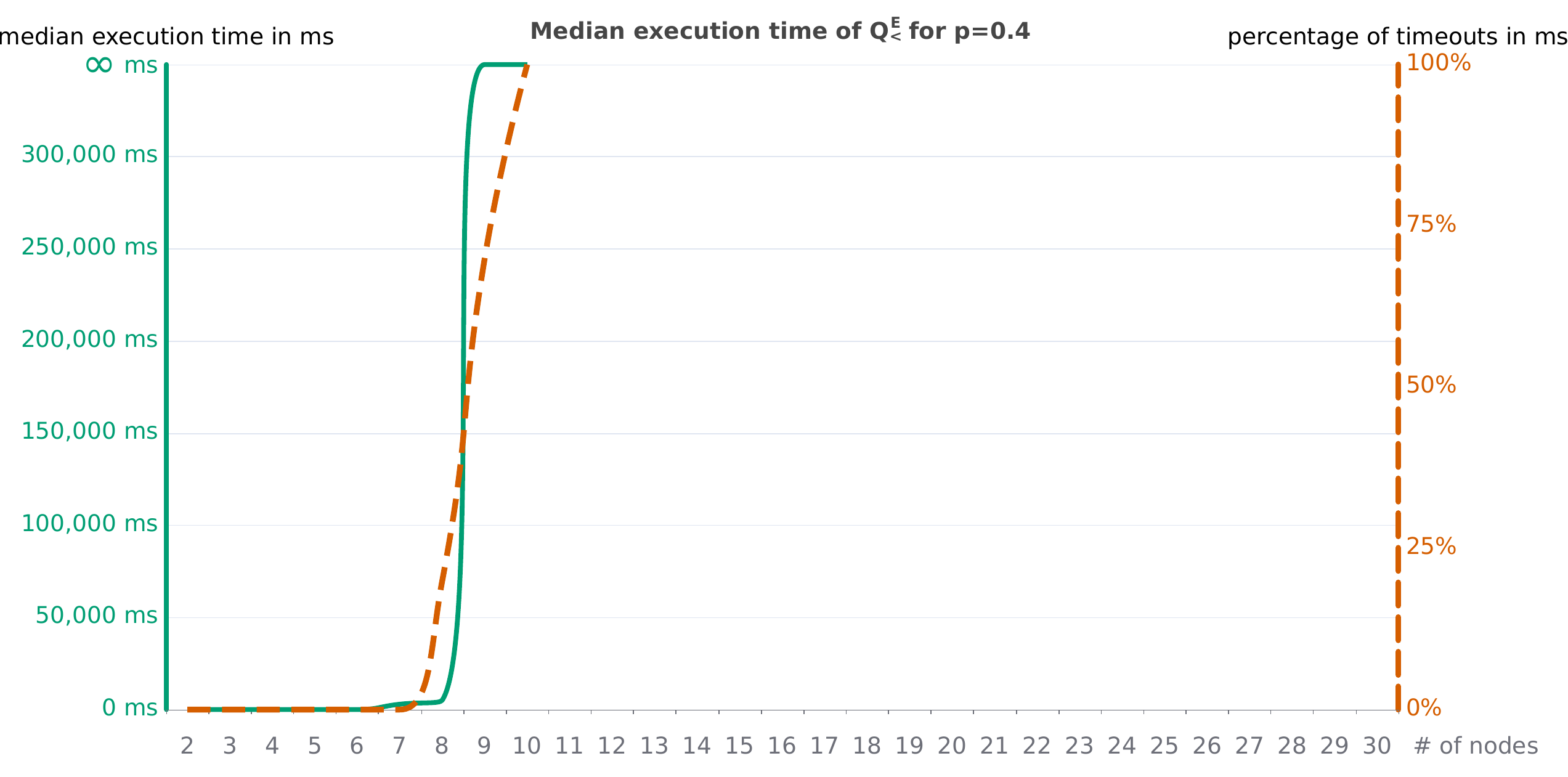}
		\caption{$p=0.4$}
		\label{neo4j-test-0.4}
	\end{subfigure}
	\begin{subfigure}{0.3\textwidth}
		\includegraphics[width=\textwidth]{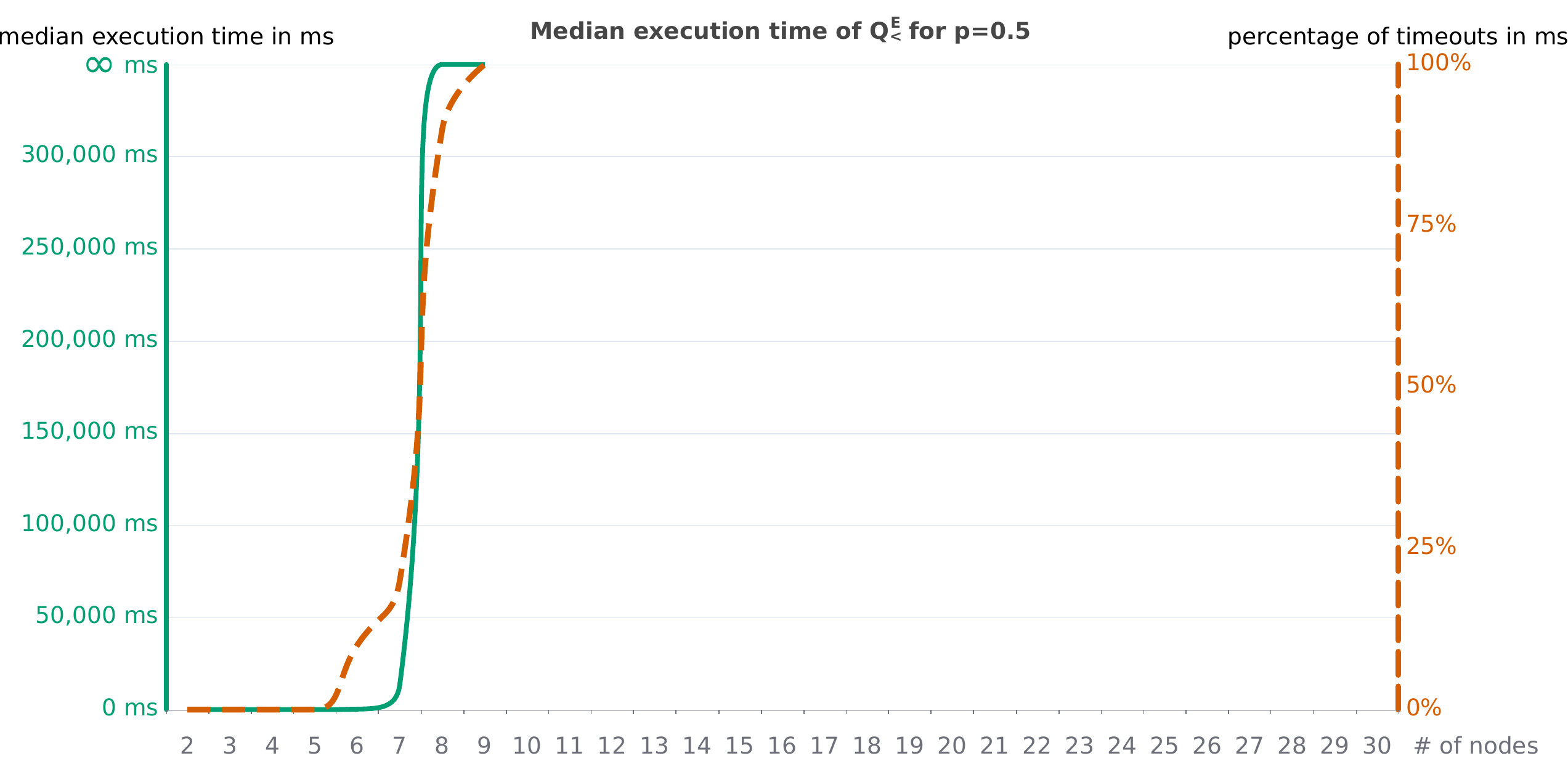}
		\caption{$p=0.5$}
		\label{neo4j-test-0.5}
	\end{subfigure}
	\caption{Timeouts and median running time of Neo4j for $\qEdgeLessThan$}
	\label{neo4j-tests}
\end{figure*}

\new{
We have shown that many queries tracing changes in property values of
edges cannot be expressed in Core GQL and PGQ, among them the simple
query $\qEdgeLessThan$ that generated significant interest in the GQL
standardization committee of ISO \cite{tobias,fred}.
Of course
real-life languages have more expressiveness than their theoretical
counterparts, and
thus real-life GQL, SQL/PGQ, and also Cypher can express this query.\
 However, they do so in a rather
convoluted way. We now show experimentally that this way of expressing
simple graph queries has no realistic
chance to work, as it generates enormous (exponential size)
intermediate results, and the query would not terminate even on {\em
	tiny graphs}.

The idea of expressing $\qEdgeLessThan$ is that its {\em complement} is easily
definable in Core GQL. In  GQL,  PGQ, and Cypher, 
paths can be named and
output.
This is an advanced feature that we omitted in our
core language: since 
paths are represented as lists, it results in non-flat outputs.
\OMIT{
	This gives a backdoor to highly inefficient expressivity of
	the queries studied in this section. The key observation is that their
	{\em complements} are expressible. For example, the complement of
	$\qdisjnodes$ is the query that finds endpoints of paths on which two
	values of property $k$ are equal. This is very easy to express by a pattern
	$\pat_{=}$, and the entire query can be formulated as the
	difference of $\big(p\df ((x) {\rightarrow} (z))^{1..\infty}\big)_G
	- \big(p\df\pat_{=}\big)_G$ which results in all the paths in a graph
	satisfying $\qdisjnodes$. Note that the first subquery simply returns
	{\em all} paths, making it highly inefficient.
}
}
\new{
The complement of  $\qEdgeLessThan$ is expressed by the pattern
$\pat^\neg$:}
\begin{gql}
(v1)->*(-[x]-> -[y]-> WHERE x.k>=y.k)->*(v2)
\end{gql}
\new{
testing for a pair of consecutive edges that break the increasing
motif; it can also be expressed by CoreGQL as
$$\pat^\neg \df (v_1) \to^{0..\infty} \big(\ () \overset{x}{\rightarrow}
() \overset{y}{\rightarrow} ()\ \big)\langle x.k \geq y.k\rangle \to^{0..\infty}(v_2)\,.$$
\OMIT{contains patterns that have a subpattern $\big(\ () \overset{x}{\rightarrow}
	() \overset{y}{\rightarrow} ()\ \big)\langle x.k \geq y.k\rangle$
	which
}
Thus, the query below
}
\begin{gql}[escapeinside={(*}{*)}]
  MATCH p = (v1) ->* (v2) RETURN v1, v2, p
  EXCEPT
  MATCH p=(*$\phi^\neg$*)  RETURN v1, v2, p
\end{gql}  
\new{finds all nodes $v_1, v_2$ and path $p$ between them satisfying
$\qEdgeLessThan$.
}
The mere fact of expressing something does not yet mean it will work
-- for example, despite SQL having the capability to simulate Turing
machines, we do not expect it to perform well with complex graph
algorithms. Likewise here, the first subquery enumerates paths between
two different nodes, and even the number of simple paths between two
nodes in a graph can grow as fast as $O(n!/n^2)$.

We now show experimentally that queries of this kind have no chance to work even on
very small graphs. A small obstacle is that there is not yet any
available implementation of GQL, and Cypher, the closest language,
chose not to have $\sqlkw{EXCEPT}$. However, there
is a way around it in Cypher by using list functions that can detect
the violation of the ``value in edges increases'' condition: 

\begin{gql}
	MATCH p=()-[*2..]->()
	WITH p, reduce(acc=relationships(p)[0].val, 
	v in relationships(p) | 
    CASE  WHEN acc=-1 THEN -1 
    WHEN v.val>=acc THEN v.val ELSE -1	END) AS inc
	WHERE NOT inc = -1 RETURN p
\end{gql}

We then tested its median running time in Neo4j\footnote{The testing
program is written in Go and communicates with Neo4j (v5.18.1) via the
Neo4j Go driver. All tests were executed on a machine with the
following configuration: 16 Intel i7-10700 @ 2.90GHz CPUs, 16GB RAM,
Ubuntu 22.04.3 LTS}, as well as the percentage of queries that time
out (the timeout is set at 300 seconds). When more than 50\% of
queries time out, the median is shown as $\infty$ms. Otherwise the
running time is computed as the median over 10 different graphs, with
1 run per graph (and an additional run 0 to
generate the appropriate indices).

The graphs on which we tested the query are the random graphs
$G(n,p)$ \cite{bollobas} on $n$ nodes where an edge exists between two
nodes with probability $p$. We considered  values $p$ between 0.1 and
0.5, with the step 0.1. As for data values in edges, they are also
randomly generated, between 0 and 100.
The reason we used this simple model of synthetic property graphs is
that it very convincingly demonstrates that the above implementation
of  $\qEdgeLessThan$ has no chance to work in practice. Even with the
smallest probability of 0.1, with a mere 24 nodes timeout was observed
in more than 50\% of all cases, and with just 30 nodes in {\em all}
cases. As the probability $p$ increases (meaning that there are more
edges in the graph, and thus the number of paths increases), the
cutoff for everything-times-out dropped to fewer than 10 nodes! 

\new{
We add two notes here. Queries $\qdisjnodes$ and $\qdisjedges$ can be similarly expressed, with the bottleneck query matching all paths being identical to the above, hence resulting in a very similar behavior. But behavior is {\em not} observed with  $\qNodeLessThan$ as it can be  
expressed directly with Cypher and GQL pattern matching, without having to find an exponential number of paths.
}

\new{
\begin{figure*}
\new{    \begin{tabular}{l|ccc|ccc|ccc|}
     & \multicolumn{3}{c}{Neo4j} & \multicolumn{3}{|c|}{Memgraph} & \multicolumn{3}{c|}{DuckDB} \\ 
     & p=0.1 & p=0.3 & p=0.5 & p=0.1 & p=0.3 & p=0.5 & p=0.1 & p=0.3 & p=0.5 \\ 
     \# nodes at 50\% timeout & 24 & 12 & 7 & 26 & 12 & 8 & N/A & N/A & N/A \\ 
     \# nodes at 100\% timeout / OOM & 30 & 13 & 8 & 28 & 13 & 9 & 164 & 131 & 128 \\
\end{tabular} }
\caption{Performance comparison of Neo4j, Memgraph, and DuckDB}
\label{comparison-fig}
\end{figure*}

}

\new{

\paragraph{Other systems}

Even though it appears that bypassing expressivity bounds by
generating exponentially many queries should not have a chance to
work (as our Neo4j tests confirm), we wanted to completely
exclude the possibility that this behavior could be due to one
specific implementation. Thus, to confirm the above results, we ran
the same tests on two other systems:
\href{https://memgraph.com}{Memgraph}, a graph-only database that uses
Cypher as its query language, and \href{https://duckdb.org}{DuckDB}, a
relational database that implements SQL/PGQ as an extension~\cite{duckpgq}.

The results confirm that the problem is the way the query is written
rather than a particular implementation. In
fact, the limits of Neo4j and Memgraph are almost identical 
as shown in Fig.~\ref{comparison-fig} which reports numbers of nodes at which 50\% and 100\% of runs timeout. 
\OMIT{: the first
system reaches 50\% timeout on 24 nodes for $p=0.1$, 12 nodes for
$p=0.3$, and 7 nodes for $p=0.1$, while the second reaches 50\%
timeout on 26 nodes for $p=0.1$, 12 nodes for $p=0.3$, and 8 nodes for
$p=0.1$. This similarity is also maintained for 100\%
timeout.
}
However, it must be noted that for the configurations that do
not timeout, the performance of Memgraph appears to be much more
efficient than that of Neo4j, with almost all test cases taking less
than 1 ms. This might be explained by the fact that the
timing procedures are different for the two systems as Memgraph does
not make query execution time available to the driver.

The performance of DuckDB was tested using the SQL/PGQ query below;
it first finds the shortest paths in the graph using the PGQ
pattern matching syntax (the inner query), then checks, in SQL, that
the edge weights appear in sorted order (the outer query).}

\begin{sql}
WITH q1 AS (SELECT *, unnest(path_edges) AS e_id
    FROM GRAPH_TABLE (testgraph
      MATCH 
      p = ANY SHORTEST (n1:N)-[e:E]->{2,}(n2:N)
      COLUMNS (edges(p) AS path_edges) ),
    LATERAL (SELECT edges.weight, 
                    edges.row_id FROM edges) ) 
SELECT path_edges, array_agg(weight) AS weights 
FROM q1 WHERE e_id=row_id GROUP BY path_edges
HAVING ARRAY_AGG(weight) = 
       ARRAY_SORT(ARRAY_AGG(weight));    
\end{sql}

\new{

While the results for DuckDB appear to be better than for native graph systems (it can handle 164 nodes with the lowest probability  $p=0.1$ before running out of memory), there is a simple explanation for it:
the difference in {\em path semantics}. The only path
semantics available for repeated patterns in DuckDB is
\textit{shortest path}, 
whereas the only one available in Neo4j and Memgraph is
\textit{trail}, which matches paths that do not go through the same edge twice. In most graphs there are significantly more trails than
shortest paths between any given nodes, hence the number of candidate
paths to be checked is much higher for Neo4j and Memgraph than for  DuckDB. Even with this, however, the way of
bypassing expressivity bounds by generating a large number of paths
can only handle very small graphs, not even reaching 200 nodes. 

}

%% file: sec-conclusions.tex
\new{A common way of enhancing the capabilities of query languages is:
(1) an initial design is produced; (2) user-demanded queries that are
not expressible in this initial design are identified; (3) their
inexpressibility is formally confirmed; (4) language deficiencies that
led to (3) are identified; (5) these deficiencies are
fixed. Note that (5) is often an iterative process that involves a deep
analysis of possible language design enhancements. Returning to
the example of recursive queries in SQL, note that while it was quite
clear that programmers cannot write queries such as transitive
closure, it took a significant effort to confirm this formally.
Subsequently hundreds of papers analyzed the expressiveness and
complexity of datalog and fixed-point extensions of relational
calculus, until the linear datalog approach was adopted by the SQL standard.

Where does this workflow put us with respect to the new
graph standards, SQL/PGQ and GQL? Prior to this paper we were at stage
(2); the results of this paper complete step (3) and put us at the
beginning of stage (4). Thus, in this last section we outline further
developments as we envision them now: what language deficiencies cause
limited expressiveness, and how they 
could be addressed. 

The main limitation of current graph query languages is that they {\em lack
compositionality}. They are essentially
graphs-to-relations languages, as opposed to graph-to-graph languages (though some ad hoc functionalities exist for
viewing relational outputs into graphs, such as the Neo4j browser  or its data science library
facilities  \cite{HodlerN22} that can extract graphs from relations  for analytics tasks; also the study of graph views is in its
infancy \cite{ives-sigmod24}). 
The lack of compositionality also manifests itself in the information flow in graph query languages. It happens in one
direction, from graphs to relations: patterns turn
graphs into relations, but there are no natural ways of going in the
other direction, and use results of relational operations of either $\RA$
or $\LRA$ to create new graphs. Of course real-life languages with
their extended functionalities provide such ways (e.g., via libraries
as mentioned above, or by writing data and using it to create new
graph views in PGQ) -- but the main point is that these lie {\em
outside} the realm of a declarative query language.

With the understanding of causes of limitations, the focus needs to be
shifted to ways to remedy those, i.e., step (5) above. If the story
of SQL is any indication, this will require a considerable research
effort. Still, we can outline what we think are possible approaches to
fixing the non-compositionality issues.

\paragraph{Pattern matching restricted to matched paths}
While graph-to-graph languages are the ultimate long-term goal, in the
short term some of non-compositionality can be fixed by letting
pattern matching  operate on previously matched patterns. We explain
this idea by the query $\qEdgeLessThan$. To find such paths from
$\ell_1$-labeled nodes to $\ell_2$-labeled nodes, we can first 
match paths from $\ell_1$-nodes to $\ell_2$-nodes and then exclude
those where the ``increasing value in edges'' condition is
violated. Specifically, we start with the pattern
$\pat \df (x) \to^{1..\infty}
(y) \langle{\ell_1(x) \wedge \ell_2(y)}\rangle$ and then create a new
pattern
\begin{equation}
\label{eq-for-concl}
\pat \ | \ \neg\exists(\overset{u}{\to} \ () \ \overset{v}{\to}\langle
u.k \geq v.k\rangle)
\end{equation}
The pipe operator $|$ means that the pattern $\overset{u}{\to} () \overset{v}{\to}\langle
u.k \geq v.k\rangle$ is evaluated on the result of the match of
pattern $\pat$, and the condition $\neg\exists$ says that there are no
matches for it. This will ensure that in the match for $\pat$ there
are no consecutive edges for which the value of their property does
not increase. 

\paragraph{Constructing graph elements from tables} 
In current graph languages the flow of information is in the direction
from graph to constructed relations. What if we could reverse it as
well and construct new graph elements from relations? As an example of
this, suppose we could construct new nodes with label \emph{old\_edge}
whose ids are edge ids in a graph $G$, and a new edge with a
lable \emph{new\_edge} which connects two such nodes coming from edges
$e_1$ and $e_2$ of $G$ if the target of $e_1$ is the source of
$e_2$. Then $$(x) \big( (u) \overset{e}{\to} (v)\langle
u.k<v.k \wedge \textit{new\_edge}(e)\rangle \big) (y)$$
(i.e., $\qNodeLessThan$ on the constructed graph) expresses
$\qEdgeLessThan$ on $G$. This powerful idea was already present in an
early theoretical language of \cite{crpq} but never properly explored
in the context of practical graph languages. 

\smallskip
Language design is a delicate process
in which the right balance must be struck between expressivity and
complexity. It is not  simple ``add these features''; 
much new research is required as they come with complexity
consequences. To give one example, in (\ref{eq-for-concl}) replace
$\overset{u}{\to} () \overset{v}{\to}\langle
u.k \geq v.k\rangle$ with $(u) \to^{1..\infty} (v)\langle u.k = v.k\rangle$. This
results in 
query $\qdisjnodes$ known to be
NP-hard in data complexity \cite{gxpath-jacm}. Thus, such extensions
are very sensitive to small syntactic changes. This however should
simply be viewed as an invitation to start a research investigation
into extensions of GQL and PGQ that allow desirable queries
without unmanageable computational overhead. Once again, the past
history of SQL tells us that such research could be very productive
and have a great influence on the language. 

}

%% file: appendix.tex
\section{Proof of theorem \ref{prop:SLRA-vs-LRA}}

At first, observe that simple linear clauses (without
$\{\query\}$) can only express conjunctive queries, as their semantics
applies operations $\join, \pi, \sigma$ and renaming to base
relations. Hence, queries of sLCRA are Boolean combinations of
conjunctive queries, known as BCCQs. While it appears to be folklore
that BCCQs are strictly contained in first-order logic, we were unable
to find the simple proof explicitly stated in the literature, hence we offer ohe
here.

Consider a vocabulary of a single unary predicate $U$ and databases
$\db_1$ and $\db_2$ such that  $\db_1(U) = \{a_1\}$ and
$\db_2(U) = \{a_1,a_2\}$ for two different constants $a_1$ and $a_2$. Since
$\db_1$ and $\db_2$ are homomorphically equivalent, they agree on all
conjunctive queries, and therefore on all BCCQs, but they do not agree
on first-order (and hence $\RA$) query that checks if relation $U$ has
exactly one element.

\section{Proof of theorem \ref{thm:ravslra}}

We first prove that $\RA(\localsch)$ is subsumed by $\LRA(\localsch)$.
Let $Q$ be a query in $\RA(\localsch)$. We show that there exists an equivalent query $Q^{\LRA}$ in $\LRA(\localsch)$, i.e. such that for every database $\db$ over $\localsch$ it holds that $\sem{Q}_{\db} = \sem{Q_L}_{\db}$ by induction on the structure of $Q$. 
\begin{itemize}
    \item (base case) If $Q = R$ then $Q_L = S$ where $S = \{R\}$
    \item If $Q = \pat_{\bar{A}} (Q')$ and $Q'_L$ is the LRA query equivalent to $Q'$ then $Q_L = Q'_L~\pi_{\bar{A}}$
    \item If $Q = \sigma_{\theta}(Q')$ and $Q'_L$ is the LRA query equivalent to $Q'$ then $Q_L = Q'_L ~ \sigma_{\theta}$
    \item If $Q = \rho_{\bar{A} \rightarrow \bar {B}}(Q')$ and and $Q'_L$ is the LRA query equivalent to $Q'$ then $Q_L = Q'_L~\rho_{\bar{A} \rightarrow \bar {B}}$
    \item[] For the three following cases, $Q'$ (resp. $Q''$) is systematically translated as $\{Q'_L\}$ (resp. $\{Q_L''\}$) 
    \item If $Q = Q' \cup Q''$ and $Q'_L$ (resp. $Q''_L$) is the LRA query equivalent to $Q'$ (resp. $Q''$) then $Q_L = \{Q'_L\} \cup \{Q''_L\}$
    \item If $Q = Q' \setminus Q''$ and $Q'_L$ (resp. $Q''_L$) is the LRA query equivalent to $Q'$ (resp. $Q''$) then $Q_L = \{Q'_L\} \setminus \{Q''_L\}$
    \item If $Q = Q' \join Q''$ and $Q'_L$ (resp. $Q''_L$) is the LRA query equivalent to $Q'$ (resp. $Q''$) then $Q_L = \{Q'_L\} \{Q''_L\}$ as, by definition of RA, it must be that $\attr{Q'} \cap \attr{Q''} = \emptyset$
\end{itemize}

We now prove that $\LRA(\localsch)$ is subsumed by $\RA(\localsch)$. Let $Q$ be a query in $\LRA(\localsch)$. We show that there exists a query $Q^{RA} \in \RA(\localsch)$ such that for every database $\db$ over $\localsch$ it holds that $\sem{Q}_{\db} = \sem{Q^{RA}}_{\db}$ by induction on the structure of $Q$.

We start with the linear clauses. Since their structure is not tree-shaped, but indeed \textit{linear}, the query must be parsed from left to right, while the equivalent $\RA$ query will be built bottom-up. Let $C = c_0, \ldots, c_n$ be a linear clause in $\LRA(\localsch)$. We build the induction on the number of clauses \mbox{of $C$.} 

\begin{itemize}
    \item (base case) If $n=1$ and $C$ is an instance of the $S$ rule, i.e. $C = R$ for some $R$ then the $\RA$ query equivalent to $C$ is $R$
    \item (base case) If $n=1$ and $C$ is one of $\pi_{\bar{A}}, \sigma_{\theta}$ or $\rho_{\bar{A} \rightarrow \bar {B}}$ then the $\RA$ query equivalent to $C$ is $I_{\emptyset}$ (the empty tuple relation)
    \item (base case) If $n=1$ and is $C$ is an instance of the $\{Q\}$ rule, then the $\RA$ query equivalent to $C$ is the one equivalent to $Q$ (see induction for queries below)
    \item[] For the inductive cases, let $C_i^{RA}$ be the $\RA$ query such that $\sem{c_0, \ldots, c_i} = \sem{C_i^{RA}}$
    \item If $n > 1$ and $c_{i+1}$ is an instance of the $S$ rule, i.e. $c_{i+1} = R$, then $C_{i+1}^{RA} = C_i^{RA} \join R$ 
    \item If $n > 1$ and $c_{i+1} = \pi_{\bar{A}}$, then $C_{i+1}^{RA} = \pi_{\bar{A}} (C_i^{RA})$
    \item If $n > 1$ and $c_{i+1} = \sigma_{\theta}$, then $C_{i+1}^{RA} = \sigma_{\theta} (C_i^{RA})$
    \item If $n > 1$ and $c_{i+1} = \rho_{\bar{A} \rightarrow \bar{B}}$, then $C_{i+1}^{RA} = \rho_{\bar{A} \rightarrow \bar{B}} (C_i^{RA})$
    \item If $n > 1$, $c_{i+1} = \{Q\}$ and $C_Q^{RA}$ is the $\RA$ query equivalent to $Q$, then $C_{i+1} = C_i^{RA} \join C_Q^{RA}$
\end{itemize}

We now treat the queries. Let $Q$ be an $\LRA$ query. We build the induction on the structure of $Q$. 

\begin{itemize}
    \item (base case) If $Q = L$ then the $\RA$ query equivalent to $Q$ is the one equivalent to $L$ (see above for induction on linear clauses).
    \item If $Q = Q_1 \cap Q_2$ and $Q_1^{RA}$ (resp. $Q_2^{RA}$) is the $\RA$ query equivalent to $Q_1$ (resp. $Q_2$) then the $\RA$ query equivalent to $Q$ is $Q_1^{RA} \cap Q_2^{RA}$
    \item If $Q = Q_1 \cup Q_2$ and $Q_1^{RA}$ (resp. $Q_2^{RA}$) is the $\RA$ query equivalent to $Q_1$ (resp. $Q_2$) then the $\RA$ query equivalent to $Q$ is $Q_1^{RA} \cup Q_2^{RA}$
    \item If $Q = Q_1 \setminus Q_2$ and $Q_1^{RA}$ (resp. $Q_2^{RA}$) is the $\RA$ query equivalent to $Q_1$ (resp. $Q_2$) then the $\RA$ query equivalent to $Q$ is $Q_1^{RA} \setminus Q_2^{RA}$
\end{itemize}

\section{Proof of Theorem~\ref{thm:Q<edge} }\label{:annex proof:Q<edge}

For an annotated path  $p$, we define $w_p$ as the sequence $e_0.k\cdots e_{n-1}.k$.


\begin{lemma}\label{lem:pathinvariance}
For every one-way path pattern $\pat$ and annotated paths $p,p'$ with $w_{p'}=w_p$, 
if $(p,\mu)\in \sem{\pat}_{p}$ then there is $\mu'$ such that $ (p',\mu')\in \sem{\pat}_{p'}$.    
\end{lemma}
\begin{proof}
    Assume that $(p,\mu)\in \sem{\pat}_{p}$. 
    Let us denote $p$ by
    \begin{center}
    \begin{tikzpicture}
    \node at (0,0) [shape=circle, draw] (node1) {$v_0$};
    \node at (1.5,0) [shape=circle] (invisiNode) {\ldots};
    \node at (3,0) [shape=circle, draw] (node2) {$v_n$};
    \draw[->] (node1.east) -- (invisiNode.west) node [above,align=center,midway]{$e_0$};
    \draw[->] (invisiNode.east) -- (node2.west) node [above,align=center,midway]{$e_{n-1}$};
\end{tikzpicture}
\end{center}
and $p'$
by 
\begin{center}
    \begin{tikzpicture}
    \node at (0,0) [shape=circle, draw] (node1) {$v'_0$};
    \node at (1.5,0) [shape=circle] (invisiNode) {\ldots};
    \node at (3,0) [shape=circle, draw] (node2) {$v'_n$};
    \draw[->] (node1.east) -- (invisiNode.west) node [above,align=center,midway]{$e'_0$};
    \draw[->] (invisiNode.east) -- (node2.west) node [above,align=center,midway]{$e'_{n-1}$};
\end{tikzpicture}
\end{center}
(Note that they have the same length since $w_p = w_{p'}$.)
    We denote by $f$ the mapping defined by $f(v_i) \df v'_i$ for $0\le i \le n$, and $f(e_i) = e'_i$ for $0 \le i \le n-1$. We then define $\mu'$ by setting $\dom{\mu'} = \dom{\mu}$ and $\mu'(x) \df f(\mu(x))$. It can be shown by induction on $\pat$ that if $(p,\mu)\in \sem{\pat}_{p}$ then  $ (p',\mu')\in \sem{\pat}_{p'}$.
\end{proof}

For two annotated paths $p\df$
    \begin{tikzpicture}
    \node at (0,0) [shape=circle, draw] (node1) {$v_0$};
    \node at (1.5,0) [shape=circle] (invisiNode) {\ldots};
    \node at (3,0) [shape=circle, draw] (node2) {$v_n$};
    \draw[->] (node1.east) -- (invisiNode.west) node [above,align=center,midway]{$e_0$};
    \draw[->] (invisiNode.east) -- (node2.west) node [above,align=center,midway]{$e_{n-1}$};
\end{tikzpicture}
and $p'\df$
    \begin{tikzpicture}
    \node at (0,0) [shape=circle, draw] (node1) {$v'_0$};
    \node at (1.5,0) [shape=circle] (invisiNode) {\ldots};
    \node at (3,0) [shape=circle, draw] (node2) {$v'_{n'}$};
    \draw[->] (node1.east) -- (invisiNode.west) node [above,align=center,midway]{$e'_0$};
    \draw[->] (invisiNode.east) -- (node2.west) node [above,align=center,midway]{$e'_{n'-1}$};
\end{tikzpicture}
 we use
${p'}\sqsubseteq p$ to denote that $p'$ is contained within $p$. Formally, 
 ${p'}\sqsubseteq p$ if there is $j$ such that 
 $v'_0 = v_j, \ldots, v'_{n'} = v_{j+n'}$, and $e'_0 = e_j,\ldots , e'_{n'-1} = e_{j+n'-1}$.
\begin{lemma}\label{lem:pathtrim}
For every one-way path pattern $\pat$ and annotated paths $p,p'$ where  ${p'}\sqsubseteq p$,   
\begin{itemize}
    \item[$(1)$] 
    if there is $\mu$ such that $(p',\mu)\in \sem{\pat}_p$ then 
$(p',\mu)\in \sem{\pat}_{p'}$
    \item[$(2)$] 
    $\sem{\pat}_{p'} \subseteq \sem{\pat}_p$
\end{itemize}
\end{lemma}
\begin{proof}
    We use the same notation as above and show the claim by a mutual induction on $\pat$. We skip the induction basis since it is trivial, and refer to the following interesting case in the induction step:
    \begin{itemize}
        \item[$(1)$] 
        \subitem $\pat = \pat_1\pat_2$ and assume $(p',\mu) \in \sem{\pat_1\pat_2}_p$. By definition, 
        there are $p_1,p_2,\mu_1,\mu_2$ such that $p' = p_1p_2$, $(p_1,\mu_1) \in \sem{\pat_1}_p, (p_2,\mu_2) \in \sem{\pat_2}_p$, $\mu_1\sim \mu_2$, and $p_2$  concatenates to $p_1$.
       Notice that $p_1,p_2\sqsubseteq p'$, and hence also $p_1p_2 \sqsubseteq p'$. By applying induction hypothesis $(1)$, we get $(p_1,\mu_1) \in \sem{\pat_1}_{p_1}, (p_2,\mu_2) \in \sem{\pat_2}_{p_2}$. By applying induction hypothesis $(2)$ we get  $(p_1,\mu_1) \in \sem{\pat_1}_{p_1p_2}, (p_2,\mu_2) \in \sem{\pat_2}_{p_1p_2}$. This allows us to conclude that, by definition,    $(p_1p_2,\mu_1\join\mu_2) \in \sem{\pat_1\pat_2}_{p_1p_2}$.
       \item[$(2)$]
         If $\pat = \pat_1 \pat_2$ then $\sem{ \pat_1 \pat_2}_{p'} =  \left\{ ({p_1\, p_2}, \mu_1\join \mu_2 ) 
    \,\middle| \begin{array}{l}(p_1,\mu_1)\in\sem{\pat_1}_{p'}, (p_2,\mu_2)\in\sem{\pat_2}_{p'},\\ \mu_1\sim \mu_2, p_2 \text{ concatenates to }p_1
    \end{array}
    \right\}$. By induction hypothesis $(2)$, $\sem{\pat_1}_{p'} \subseteq \sem{\pat_1}_{p}$ and $\sem{\pat_2}_{p'} \subseteq \sem{\pat_2}_{p}$. Thus,  
    \[\sem{ \pat_1 \pat_2}_{p'} \subseteq  \left\{ ({p_1\, p_2}, \mu_1\join \mu_2 ) 
    \,\middle| \begin{array}{l}(p_1,\mu_1)\in\sem{\pat_1}_{p}, (p_2,\mu_2)\in\sem{\pat_2}_{p},\\ \mu_1\sim \mu_2, p_2 \text{ concatenates to }p_1
    \end{array}
    \right\}.\] By definition,  $\sem{ \pat_1 \pat_2}_{p'} \subseteq \sem{ \pat_1 \pat_2}_{p}$.
    \end{itemize}
\end{proof}

\OMIT{
\begin{lemma}\label{lem:pathextend}
For every path-pattern $\pat$, and word paths $p',p$ where
   $p' \sqsubseteq p$
     it holds that  $\sem{\pat}_{p'} \subseteq \sem{\pat}_p$.
\end{lemma}
\begin{proof}
    We show the claim by induction on $\pat$.
    
    \textbf{Induction Base:} \todo{complete}
    
    \textbf{Induction Step:} \todo{complete}

\end{proof}
}

Let $\pat$ be a path pattern. We define 
\[
\lang{\pat} = \{w \mid \forall\,p: \left(w=w_p \rightarrow \exists\, \mu:(p,\mu )\in \sem{\pat}_{p} \right) \}
\]
\OMIT{say that a word $w \in \lang{\pat}$ if for every annotated path $p$ where $w=w_p$ there is a partial mapping $\mu$ such that $(p,\mu )\in \sem{\pat}_{p}$.
}
For concatenation, we have:
\begin{lemma}\label{lem:breakconcat}
For every one-way path patterns 
$\pat_1,\pat_2$, the following equivalence holds:
$\lang{\pat_1 \pat_2} = \lang{\pat_1 }\lang{\pat_2}$.
\end{lemma}
\begin{proof}

\textbf{$\subseteq$ direction:}
Assume $w\in\lang{\pat_1 \pat_2}$. By definition, for every $p$ such that $w=w_p$ there is $\mu $ such  that $(p,\mu) \in \sem{\pat_1\pat_2}_p$.
In turn, by definition of $\sem{\pat_1\pat_2}_p$ we can conclude that there are paths $p_1,p_2$ such that $p=p_1p_2$ and there are $\mu_1,\mu_2$ for which $(p_1,\mu_1)\in \sem{\pat_1}_p$, $(p_2,\mu_2)\in \sem{\pat_2}_p$, and $\mu_1\sim \mu_2$. We set $w_1 = w_{p_1}$ and $w_2 = w_{p_2}$. For $p_1$ we already saw that there is a $\mu_1$ such that $(p_1,\mu_1)\in \sem{\pat_1}_p$. Let $p'$ be such that $w_1 = w_{p_1} = w_{p'}$. By Lemma~\ref{lem:pathinvariance}, we can conclude that there is $\mu'$ such that $(p',\mu')\in\sem{\pat_1}_{p'}$. Hence, by definition $w_1 \in \lang{\pat_1}$. We can show similarly that $w_2 \in \lang{\pat_2}$, which completes this direction. 

\textbf{$\supseteq$ direction:} Assume that 
$w\in \lang{\pat_1 }\lang{\pat_2}$. That is, there are 
$w_1, w_2$ such that $w= w_1w_2$ and $w_1 \in \lang{\pat_1}, w_2\in \lang{\pat_2}$.
Let $p$ be a path such that $w_p = w = w_1w_2$. Let $p_1,p_2$ be its subpaths such that $p=p_1p_2$ and $w_{p_1}=w_1, w_{p_2}= w_2$.
Since $w_{p_1}=w_1$ there are $\mu_1,\mu_2$ such that $(p_1,\mu_1)\in \sem{\pat_1}_{p_1}, (p_2,\mu_2)\in \sem{\pat_2}_{p_2}$.
Since $\pat$ is a one-way path pattern, it holds that $\dom{\mu_1} \cap \dom{\mu_2} = \emptyset$, and thus $\mu_1 \sim \mu_2$.
Applying Lemma~\ref{lem:pathtrim} $(2)$ enables us to conclude that  $(p_1,\mu_1)\in \sem{\pat_1}_{p}, (p_2,\mu_2)\in \sem{\pat_2}_{p}$ since $p_1,p_2\sqsubseteq p$. 
By definition of $\sem{\pat_1\pat_2}_p$, we can conclude that $(p_1p_2,\mu_1\join \mu_2)\in\sem{\pat_1\pat_2}_p$. 
This suffices to conclude that $w\in \lang{\pat_1\pat_2}$, which completes this direction. 
\end{proof}

For repetition, we have:
\begin{lemma}
\label{lem:rep}
For every one-way path pattern $\pat$,
$\mathcal{L}(\pat^2)= \mathcal{L}(\pat) \mathcal{L}(\pat)$.
\end{lemma}
\begin{proof}

\textbf{$\subseteq$ direction:}
    Assume that $w\in \lang{\pat^2}$ and let $p$ be a path with $w=w_p$.
    By definition 
    $\sem{\pat^2}_p = \{(p_1p_2,\emptyset) \mid \exists \mu_1,\mu_2:\,(p_1,\mu_1) \in \sem{\pat}_p , (p_2,\mu_2) \in \sem{\pat}_p ,\text{ and $p_2$ concatenates to $p_1$}  \}$. Therefore, $p = p_1p_2$ for some $p_1,p_2$ such that there are $\mu_1,\mu_2$ where $(p_1,\mu_1), (p_2,\mu_2) \in \sem{\pat}_p$.
    Let us denote $w_1 \df w_{p_1}$ and $w_2 \df w_{p_2}$. It holds that $w=w_1w_2$ and it suffices to show that $w_1,w_2\in \lang{\pat}$.
    For $p_1$ there is $\mu_1$ such that $(p_1,\mu_1 )\in \sem{\pat}_p$. Due to Lemma~\ref{lem:pathtrim} $(1)$, since $p_1\sqsubseteq p$ we can conclude that  $(p_1,\mu_1 )\in \sem{\pat}_{p_1}$. Due to Lemma~\ref{lem:pathinvariance}, for every $p$ with $w_p = w_1$ it holds that $(p,\mu_1 )\in \sem{\pat}_{p}$. Hence $w_1\in\lang{\pat}$. We show similarly that  $w_2\in\lang{\pat}$, which completes this direction.

    \textbf{$\supseteq$ direction:}
    Assume that $w_1,w_2\in \lang{\pat}$.
    Let $p_1,p_2$ be such that $w_{p_1}=w_1, w_{p_2} = w_2$. There are $\mu_1,\mu_2$ such that $(p_1,\mu_1)\in \sem{\pat}_{p_1}$ and
    $(p_2,\mu_2)\in \sem{\pat}_{p_2}$.
    If $p_2$ does not concatenate to $p_1$, then by Lemma~\ref{lem:pathinvariance} we can change its first node id to match the last node id of $p_1$. Thus, we can assume that $p_2$ concatenates to $p_1$.
    Due to Lemma~\ref{lem:pathtrim} $(2)$, it holds that  $(p_1,\mu_1)\in \sem{\pat}_{p}$ and
    $(p_2,\mu_2)\in \sem{\pat}_{p}$
    where $p=p_1 p_2$. Therefore, 
    $(p_1p_2,\emptyset) \in \sem{\pat^2}_{p}$ by definition which completes this direction.
\end{proof}

We can further generalize this lemma by replacing $2$ with any $k\ge 1$. 
\begin{lemma}\label{lem:breakrepetition}
For every one-way path patterns $\pat$,
    $\lang{\pat^k}= \left(\lang{\pat}\right)^k$, $k\ge 1$
\end{lemma}
\begin{proof}
    The claim can be shown by induction on $k$  by combining Lemmas~\ref{lem:breakconcat} and~\ref{lem:rep}.
\end{proof}



To apply the above lemmas to our settings we first present a Normal Form for path patterns. 
We say that a one-way path pattern is in $+$NF if $+$ occurs only under unbounded repetition. 
\begin{lemma}\label{lem:plusNF}
    For every one-way path pattern $\pat$ there is a one-way path pattern $\pat'$ in $+$NF such that for every annotated path $p$ it holds that $\sem{\pat}_p = \sem{\pat'}_p$. 
\end{lemma}
\begin{proof}
    We show that there are translation rules that preserve semantics on annotated paths.
    \newcommand{\tr}{\mathsf{tr}}
    This translation $\tr$ is described inductively.
    \begin{itemize}
        \item $\tr\left((x)\right)\df (x)$
        \item 
      $ \tr \left( \overset{x}{\rightarrow}\right) \df \overset{x}{\rightarrow}$
        \item 
      $ \tr \left( \overset{x}{\leftarrow}\right) \df \overset{x}{\leftarrow}$
      \item 
       $ \tr \left( \pat_1+ \pat_2 \right) \df \pat_1+ \pat_2$
       \item 
       $\tr\left( \pat_1 \pat_2\right)
       \df 
       +_{1\le i\le k, 1\le j\le m} \rho_i \rho'_j
       $
       where  $\tr\left( \pat_1 \right) \df \rho_1 + \cdots +\rho_k$ and 
        $\tr\left( \pat_2 \right) \df \rho'_1 + \cdots +\rho'_m$ 
        \item When $m<\infty$ we define $\tr\left( \pat^{n..m}\right) \df \tr(\underbrace{\pat\cdots \pat}_{n\text{ times }}) +\cdots + \tr(\underbrace{\pat\cdots \pat}_{m\text{ times }}) $
        \item  When $m=\infty$ we define 
        $\tr\left( \pat^{n..m}\right) \df \pat^{n..m}$
        \item 
        $\tr \left(\pat_{\langle \theta \rangle} \right) \df \tr \left(\pat \right)_{\langle \theta \rangle}$
    \end{itemize}
    It can be shown that the output of $\tr$ is indeed in $+$NF and
    the equivalence of the semantics of $\pat$ and $\tr({\pat})$ can be  shown by induction based on Lemmas~\ref{lem:breakrepetition} and~\ref{lem:breakconcat}. 
\end{proof}

We are now ready to move to the main proof. 
\begin{proof}[Proof of Theorem~\ref{thm:Q<edge}]
   Let $\pat$ be a one-way path pattern in $+$NF of the form $\pat_1+\cdots+ \pat_m$.
By pigeonhole principle, if $\lang{\pat}$ is infinite then there is at least one $\pat_i$ with infinite $\lang{\pat_i}$. Since $\pat$ is in $+$NF we can denote $\pat_i$ as $\rho_1 \rho^{n..\infty} \rho_2$. 
By Lemmas~\ref{lem:breakconcat} and~\ref{lem:breakrepetition}, 
$\lang{\pat_i} = \lang{\rho_1} \lang{ \rho}^{n..\infty} \lang{\rho_2}$. Let $w$ be a word in $\lang{ \rho}^{n..\infty}$. As $w \in \lang{ \rho}^{n..\infty}$, then $ww \in \lang{ \rho}^{n..\infty}$ as well, i.e. there exists a pair $(p, \mu) \in \sem{\pat}$ such that $ww$ is a subword of $w_p$. However, as $ww$ repeats each value of $w$ twice, it cannot satisfy the condition. 
\end{proof}

\new{
\subsection{Generalizing Theorem~\ref{thm:Q<edge}}

We define generalized annotated paths as 
paths whose nodes and edges are annotated each with one data value attached to a key $k$. 
We define the language of these paths as sets of words, which we also refer to as \emph{path annotations}, that are obtained by concatenating the data values according to their position. 
In particular, each word in the language starts with a data value of the first node and ends with that of the last one. In between it alternates between data of edge and node.
 For example, the language $\lang{}$ of a generalized path is an ....
 
 We define the operator $\odot$ on path annotations $p_1,p_2$ whenever the last letter of $p_1$ equals to the first letter of $p_2$ as follows:
 \[p_1 \odot p_2 \df n_0.k\, e_0.k\, \cdots  n_{m-1}.k\, e_{m-1}.k\, n_m.k\, e'_0.k\, \cdots  n'_{\ell-1}.k\, e'_{\ell-1}.k\, n_{\ell}.k
\]
We generalize this definition to sets of path annotations in the standard way:
\[
P_1 \odot P_2 \df \{ p_1\cdot p_2 \,\vline \, p_1 \in P_1, p_2 \in P_2, p_2\text{ concatenates to } p_1 \}
\] 

\begin{lemma}\label{lem:pathlang}
	For every one way path patterns $\pat, \pat_1,\pat_2$ the following holds:
	\begin{itemize}
		\item $\lang{\pat_1\pat_2} = \lang{\pat_1} \odot \lang{\pat_2}$
		\item $\lang{\pat^n} = \underbrace{\lang{\pat} \odot \cdots \odot \lang{\pat}}_{n \text{ times}}$
		\item $\lang{\pat^{n..\infty}}= \bigcup_{k=n}^\infty \lang{\pat^k}$
	\end{itemize} 
\end{lemma}
\begin{proof}
	.....
\end{proof}

With this, we can move to prove a pumping-lemma for one-way path patterns. This lemma shows that whenever we have a path pattern that accepts paths of unbounded length, we can obtain infinitely many paths that are also accepted by the pattern. 
\begin{theorem}\label{thm:pump}
    Let $\pat$ be a one-way path pattern such that for every $n\in \mathbb{N}$ there is an annotated path $p$ of length $n$ accepted by $\psi$. Then there exists $n_0\in \mathbb{N}$ such that if $|p|>n_0$ then $p=p_1 p_2 p_3$ with $|p_2|>1$ and for every $n\in \mathbb{N}$, the annotated path $p_1 p_2^{n+1}p_3 $ is accepted by $\psi$.
\end{theorem}
\begin{proof}
   Let $\pat$ be a one-way path pattern in $+$NF of the form $\pat_1+\cdots+ \pat_m$.
Assume that for every $n\in \mathbb{N}$ there is an annotated path $p$ of length $n$ in $\lang{\psi}$.  This implies that there is at least one $\pat_i$ for which the annotated paths in $\lang{\pat_i}$ are not of bounded length. That implies also that   $\lang{\pat_i}$ is infinite. 
Since $\pat$ is in $+$NF we can
assume
by Lemma~\ref{lem:pathlang} that there are $\rho_1,\rho,\rho_2$ such that 
$\lang{\pat_i} = \lang{\rho_1} \lang{ \rho}^{n..\infty} \lang{\rho_2}$ with $ \lang{ \rho}^{n..\infty}$ of unbounded length \footnote{ Notice that it can be the case that $\lang{\pat_i} = \lang{\rho_1} \lang{ \rho}^{n..\infty} \lang{\rho_2} \langle \theta \rangle$ in which case the sequel of the proof still holds since the condition $\theta$ can only refer to variables in $\rho_1,\rho_2$.}.
This implies that not only  $\lang{ \rho}$ contains at least one path annotation with at least three symbols (which means at least one edge data value) but also that there are path annotations $p_1,\ldots ,p_k \in \lang{\rho}$ such that $p_1\odot \cdots \odot p_k$ is well defined and $p_1$ concatenates to it.
As $p_1,\ldots ,p_k  \in \lang{ \rho}^{n..\infty}$, then by definition of $\lang{\,}$ so does $\underbrace{p_1\odot\cdots \odot p_k  \odot \cdots \odot p_1 \odot \cdots \odot p_k }_{\ell+1\text{ times }} \in \lang{ \rho}^{n..\infty}$ for every $\ell \in \mathbb{N}$.
Using  Lemma~\ref{lem:pathlang}  we can now compose new path annotations in  $\lang{ \rho}^{n..\infty}$ which completes the proof of the claim.
\end{proof}

Based on this pumping argument we are now ready to prove the main inexpressibility result.
\begin{proof}[proof of Theorem~\ref{thm:Q<edge}]
 Assume by contradiction that there is an order motif $w$ for which $Q^E_w$ is expressible by a one-way path pattern $\pat$.
Since $\pat$ accepts paths of unbounded length, we can apply Therorem~\ref{thm:pump} and obtain that if $p$ is accepted by $\pat$ then so does its pumped version. However, by pumping $p$ we obtain a path for which the edges do not conform the motif $w$. \liat{Can elaborate more depending on which part of $w$ the pumped part conforms to.... }
\end{proof}

\section{Proof of Corollary~\ref{thm:QNdiff}}
The claim on $\qdisjedges$ is proved similarly to Theorem~\ref{thm:Q<edge}.
	Assume by contradiction that there is a one-way path pattern $\pat$ that expresses $\qdisjnodes$.  Using the same notation as in the proof of Theorem~\ref{thm:Q<edge} we get a contradiction since if  $p_1,\ldots,p_k$ consists of three symbols it has to be the case that the first and last are the same, and otherwise the pumping results in a path annotation that does not satisfy the condition. 
}

\section{Proof of Theorem~\ref{thm:QNdiff}}

First, notice that if a path contains a cycle, it cannot be an answer to $\qdisjnodes$ as the value in the node the path comes back to would be repeated (at least) twice, so we can safely consider only "path shaped" graphs with values on the nodes. Formally, we look at node-annotated paths of the form 
\begin{center}
    \begin{tikzpicture}
    \node at (0,0) [shape=circle, draw] (node1) {$v_0$};
    \node at (1.5,0) [shape=circle] (invisiNode) {\ldots};
    \node at (3,0) [shape=circle, draw] (node2) {$v_n$};
    \draw[->] (node1.east) -- (invisiNode.west) node [above,align=center,midway]{$e_0$};
    \draw[->] (invisiNode.east) -- (node2.west) node [above,align=center,midway]{$e_{n-1}$};
\end{tikzpicture}
\end{center}
where $v_0, \ldots, v_n$ are distinct nodes, $e_0, \ldots, e_{n-1}$ are distinct edges (for $n > 0$) and each $v_i$ has the property $v_i.k$ defined. The node-value word of a node-annotated path is defined as $w^N_p = v_1.k \cdot v_2.k \cdots v_n.k$. 

As in Section~\ref{:annex proof:Q<edge}, we define the node-value language of a one-way path pattern $\pat$ as $\mathcal{L}^N(\pat) = \{w \mid \exists (p, \mu) \in \sem{\pat}_p, w^N_p = w\}$.

\begin{lemma}\label{lem:node-value-pumping}
    For every one-way infinitely-repeated path pattern $\pat^{n..\infty}$ , if a word $w \in \mathcal{L}^N(\pat^{n..\infty})$ then the word $ww$ also belongs to $\mathcal{L}^N(\pat^{n..\infty})$.
\end{lemma}

\begin{proof}
    Let $w$ be a word in $\mathcal{L}^N(\pat^{n..\infty})$. By definition of $\mathcal{L}^N$, there exists a path $\pathval(v_1, e_1, \ldots e_{k-1}, v_k)$ and a valuation $\mu$ such that $(p, \mu) \in \sem{\pat^{n..\infty}}_p$ and $w^N_p = w$. Let $p'=\pathval(v'_1, e'_1, \ldots e'_{k-1}, v'_k)$ be a copy of $p$ with fresh ids for all elements (nodes and edges) except for $v'_1$ for which $id(v'_1) = id(v_k)$. By definition of path concatenation, $p$ and $p'$ concatenate and, since $(p, \mu) \in \sem{\pat^{n..\infty}}_p$, we can construct a valuation $\mu'$ such that $(p', \mu') \in \sem{\pat^{n..\infty}}_{p'}$ as follows: $\mu'(x)=v_i'$ whenever $\mu(x)=v_i$ for all $0 \leq i \leq k$ and $\mu'(y)=e'_j$ whenever $\mu(y)=e_j$ for all $0 \leq j < k$. As $(p, \mu) \in \sem{\pat^{n..\infty}}_p$ (resp. $(p', \mu') \in \sem{\pat^{n..\infty}}_{p'}$), it can easily be shown by induction that $(p, \mu) \in \sem{\pat^{n..\infty}}_{p\concatto p'}$ (resp. $(p', \mu') \in \sem{\pat^{n..\infty}}_{p\concatto p'}$) and, by definition of the semantics of repetition, we can conclude that $(p \concatto p',\mu_\emptyset) \in \sem{\pat^{n..\infty}}_{p\concatto p'}$ and so we get that $w^N_p \cdot w^N_{p'} = w \cdot w$ belongs to $\mathcal{L}^N(\pat^{n..\infty})$. 
\end{proof}

\begin{lemma}\label{lem:node-value-subword}
    For any one-way path  patterns $\pat, \pat', \pat''$ such that $\mathcal{L}^N(\pat' \pat \pat'') \neq \emptyset$ and any node-value word $w$, if $w \in \mathcal{L}^N(\pat)$ then there exists a node-annotated path $p$ such that $w^N_p \in \mathcal{L}^N(\pat' \pat \pat'')$ and $w$ is a subword of $w^N_p$.
\end{lemma}

\begin{proof}
    Let $w$ be a word in $\mathcal{L}^N(\pat)$. By definition of $\mathcal{L}^N$, there exists a path $p=\pathval(v_1, e_1, \ldots, v_k)$ and an assignment $\mu$ such that $(p, \mu) \in \sem{\pat}_G$ and $w = w^N_p$ for some graph $G$. By the semantics of concatenation and since $\mathcal{L}^N(\pat_1 \pat \pat_2) \neq \emptyset$, there are $p'=\pathval(v'_1, e'_1, \ldots, v'_{k'})$, $p''=\pathval(v''_1, e''_1, \ldots, v''_{k''})$, $\mu'$ and $\mu''$ such that $(p', \mu') \in \sem{\pat'}_G$, $(p'', \mu'') \in \sem{\pat''}_G$, $p' \concatto p$, $p \concatto p''$ and $\mu, \mu'$ and $\mu''$ are compatible. By definition of path concatenation, $p' \concatto p \concatto p'' = \pathval(v'_1, e'_1, \ldots, v'_{k'}, e_1, \ldots, v_k, e''_1, \ldots v''_{k''})$ forms a path whose node-value word is $w^N_{p' \concatto p \concatto p''} = v'_1.k \cdots v'_{k'}.k \cdot v_2.k \cdots v_k.k \cdot v''_2.k \cdot v''_{k''}.k$. Since the paths concatenate, we know that $v'_{k'} = v_1$ and $v''_1 = v_k$ and so $w$ is a subword of $w^N_{p' \concatto p \concatto p''}$. Once again by definition of concatenation, we get that $(p' \concatto p \concatto p'', \mu' \join \mu \join \mu'') \in \sem{\pat' \pat \pat''}_G$ and so $w^N_{p' \concatto p \concatto p''} \in \mathcal{L}^N(\pat' \pat \pat'')$. 
\end{proof}

The above lemma can be easily extended to concatenations of arbitrary size. 

We can now prove Theorem~\ref{thm:QNdiff}.

Assume, by contradiction, that there is a one-way path pattern $\pat$ equivalent to $\qdisjnodes$. By lemma~\ref{lem:plusNF}, we can assume that $\pat$ is in +NF, and so of the shape $\pat_1+\cdots+\pat_m$, and since the language of $\qdisjnodes$ is infinite, we can further assume that at least one of the $\pat_i$ is of the shape $\rho_1 \ldots \rho_j^{n..\infty} \ldots \rho_k$. Let $w$ be a word in $\mathcal{L}^N(\rho_j^{n..\infty})$. By lemma~\ref{lem:node-value-pumping}, we have that $ww$ also belongs to $\mathcal{L}^N(\rho_j^{n..\infty})$ and, by lemma~\ref{lem:node-value-subword} we can conclude that there exists a path $p$ and an assignment $\mu$ such that $ww$ is a subword of $w^N_p$ and $(p, \mu) \in \mathcal{L}^N(\rho_1 \ldots \rho_j^{n..\infty} \ldots \rho_k)$. Since all values in $ww$ are repeated twice, $\pat$ cannot be equivalent to $\qdisjnodes$.

\section{Proof of Theorem~\ref{main-gql-thm}}

\newcommand{\pos}{\mathsf{pos}}
\newcommand{\dlpath}{\mathbf{p}_\epsilon}

In what follows, we focus on \emph{data-less paths} which are, as before, graphs whose underlying structures are paths, without properties and labels. 
{\begin{center}
    \begin{tikzpicture}
    \node at (0,0) [shape=circle, draw] (node1) {$v_0:\mathsf{first}$};
    \node at (2,0) [shape=circle] (invisiNode) {\ldots};
    \node at (4,0) [shape=circle, draw] (node2) {$v_n:\mathsf{last}$};
    \draw[->] (node1.east) -- (invisiNode.west) node [above,align=center,midway]{$e_0$};
    \draw[->] (invisiNode.east) -- (node2.west) node [above,align=center,midway]{$e_{n-1}$};
\end{tikzpicture}
\end{center}
\noindent where $n\ge 0$, and $v_i \ne v_j\in\Nodeset, e_i \ne e_j\in \Edgeset$ for every $i \ne j$, and $\lbl(v_1) = \mathsf{first},\,\lbl(v_n)= \mathsf{last}$. 
}

We then show that no Boolean Core GQL query returns true on
 $\gdb_n$ iff $n$ is a power of 2.

\begin{theorem}\label{thm:2_pow_n}
There is no Core GQL query $\query$ such that for every data-less path $\mathbf{p}$ the following holds:
\[
\sem{\query}_\mathbf{p} = \true \text{ if and only if }\, 
\len(\mathbf{p})=2^n,n\in \mathbb{N}
\]
\end{theorem}

\subsection{Proof of Theorem~\ref{thm:2_pow_n}}

For a data-less path $\dlpath$, we define the function $\pos$ that maps each node in $\dlpath$ to its position (where the position of the first node in $\dlpath$ is $1$) and each edge in $\dlpath$ to the tagged position of its source node. We do so to 
distingush whether a position correspond to a node or an edge.
For example, if $\dlpath$ is \begin{tikzpicture}
    \node at (0,0) [shape=circle, draw] (node1) {$n_1$};
    \node at (1.5,0) [shape=circle] (invisiNode) {\ldots};
    \node at (3,0) [shape=circle, draw] (node2) {$n_k$};
    \draw[->] (node1.east) -- (invisiNode.west) node [above,align=center,midway]{$e_1$};
    \draw[->] (invisiNode.east) -- (node2.west) node [above,align=center,midway]{$e_{k-1}$};
\end{tikzpicture} 
then  $\pos(n_j) = j$ and $\pos(e_j) = j'$.
We treat tagged integers as integers wrt to arithmetics and equality.

Before proving the translation from path patterns to PA, we first show that path patterns without variables can be translated to automata. 

\begin{lemma}\label{lem:pat_to_aut}
    For every path pattern $\pat$ with no variables,  there exists a finite automaton 
    $A_{\pat}$  such that
    \[     o_1 \cdots o_{m-1}  \in \lang{A_{\pat}}\, \text{iff}\,
         ((n_1,e_1,\ldots,n_{m-1},e_{m-1},n_m),\emptyset)\in \sem{\pat}_{\dlpath}  
         \]
         where $o_j \df a$ if $e_j$ is a forward edge (from $n_{j-1}$ to $n_j$) and  $o_j \df b$ if $e_j$ is a backward edge (from $n_{j}$ to $n_{j-1}$).  
\end{lemma}
\newcommand{\rtrans}{\mathsf{right}}
\newcommand{\ltrans}{\mathsf{left}}

\begin{proof}
We use the standard definition of finite automata.
We fix the input alphabet to be $ \{a,b\}$, 
    and prove the claim by induction on $\pat$'s structure. 
    
    The intuition behind the construction is that the 
    $a$'s represent the forward edges and the 
    $b$'s the backward edges in a path that matches 
    $\pat$.
    
    As a convention, $q_0$ and $ q_f$ denote the initial and final states, respectively.
    \paragraph{Base cases.}
    \begin{itemize}
        \item $\pat = ()$ then the automaton uses the following transitions $\delta(q_0,\epsilon)= q_f$.
        \item  $\pat = () \rightarrow ()$ then 
      $\delta(q_0,a)= q_f$.
        \item  $\pat = () \leftarrow ()$ then   $\delta(q_0,b)= q_f$.
        \end{itemize}
    \paragraph{Induction step}
    For the cases
        $\pat = \pat_1 + \pat_2$, 
        $\pat = \pat_1 \pat_2$, and
        $\pat = \pat'^{*}$ 
    we use similar construction to Thompson construction by adding epsilon transitions and redefining the initial and accepting states. 
    It is straightforward to show that this construction satisfies the connection stated in the lemma. 
\end{proof}
\begin{corollary}\label{cor:sl}
    For every path pattern $\pat$ without variables, and every data-less path $\dlpath$ the set 
    $\{ 
    \pos(n_m) -\pos(n_1) \mid (p,\emptyset)\in \sem{\pat}_{\dlpath}, p \df(n_1, e_1,\ldots, n_{m-1},e_{m-1},n_m) \}$ is semi-linear.   
\end{corollary}
\begin{proof}
    The proof follows from Lemma~\ref{lem:pat_to_aut},  the fact that the Parikh image of a finite automaton is semilinear, and the closure of semilinear sets under subtraction.
\end{proof}

We obtain the following connection: 
\begin{lemma}\label{lem:dlp_to_pa}
    For every data-less path $\dlpath$, and path pattern $\pat$
    there is a PA formula $\phi_{\pat}(x_s,x_1,\ldots, x_m, x_t)$ where   $\sch{\pat} \df\{x_1,\ldots, x_m\}$ such that 
    \[
\phi_{\pat}(i_s,i_1,\ldots, i_m, i_t) \,\text{ if and only if }\,
  \exists p:\,(p,\mu) \in\sem{\pat}_{\dlpath} 
    \]
    where  $i_j = \pos(\mu(x_j))$ for every $j$, and $i_s,i_t$ are the positions of the first and last nodes of $p$ in $\dlpath$, respectively.
\end{lemma}
\begin{proof}
We simplify $\pat$:
Since we deal with data-less paths, it is straightforward that any $\pat_{\langle \theta \rangle}$ can be rewritten to $\pat$ while preserving the semantics (on data-less paths). Formally, for every data-less path $P$ and for every condition $\theta$, it holds that $\sem{\pat}_P = \sem{\pat_{\langle \theta \rangle}}_P$.
In addition, we can assume that every repetition $\pat^{n..m}$ where $m<\infty$ can be rewritten to 
$\cup_{i=n}^m \pat_i$. This is true for any input (no need to restrict to data-less paths) due to the semantics definition. 
For repetition $\pat^{n..m}$ with $m=\infty$, we can rewrite the pattern and obtain  $\pat^{n} \pat^{0..\infty}$. Also here, equivalence holds for any input (no need to restrict to data-less paths) due to the semantics definition. 
Recall that now we focus only on path patterns in which we disallow having variables on the edges.

    Under the above assumptions on the form of $\pat$, we define $\phi_{\pat}$ inductively as follows:
\begin{itemize}
    \item If $\pat\df (x)$ then $\phi_{\pat}(x_s,x,x_{t}) \df x_s= x \wedge x=x_{t}$;
   \item If $\pat\df ()$ then $\phi_{\pat}(x_s,x_{t}) \df x_s=x_{t}$;
    \item If $\pat\df (x)\rightarrow(y)$ then $\phi_{\pat}(x_s,x,y,x_t) \df y=x+1 \wedge x=x_s \wedge y=y_s$
  \item   If $\pat\df (x)\leftarrow(y)$ then $\phi_{\pat}(x_s,x,y,x_t) \df x=y+1 \wedge x=x_s \wedge y=y_s$
   \item If $\pat\df \overset{x}{\rightarrow}$ then $\phi_{\pat}(x_s,x,x_t) \df x_t = x_s+1   \wedge x=x_s$
    \item  
    If $\pat\df \pat_1 \pat_2$ then 
    $\phi_{\pat}(x_s,\bar{z}, y_t) \df  \phi_{\pat_1}(x_{s},\bar{x},x_t)
    \wedge \phi_{\pat_2}(y_s,\bar{y},y_t) \wedge x_t = y_s 
   $ where $\bar{z}$ is the union  $\bar{x}\cup \bar{y}$.
\item If $\pat\df \pat_1 + \pat_2$
then 
  $\phi_{\pat}(x_s,\bar{x}, x_t) \df \phi_{\pat_1}(x_s,\bar{x}, x_t) \vee 
  \phi_{\pat_2}(x_s,\bar{x}, x_t)$.
  \item  If $\pat\df \pat_1^{n}$ then
$\phi_{\pat}(x_s,x_t) = \exists x_1,y_1,\bar{z}_1\ldots,x_n, y_{n}, \bar{z}_n:\ 
      \bigwedge_{i=1}^n\phi_{\pat_1}(x_i,\bar{z}_i,y_i)
      \wedge 
      \bigwedge_{i=1}^{n-1} y_{i}=x_{i+1}
      \wedge
      x_1 = x_s \wedge y_n = x_t
      $
  \item  If $\pat\df \pat_1^{*}$ then,
due to Corollary~\ref{cor:sl}, there is a PA formula 
$\psi_{\pat_1^{*}}(m)$ for $\{ |p| \mid (p,\emptyset) \in \sem{\pat_1^{*}}_{\pm} \}$. Therefore,
$\phi_{\pat}(x_s,x_t) \df \exists m: x_t = x_s+ \psi_{\pat_1^{*}}(m)$ 
\end{itemize}
Showing that the condition in the Lemma holds for $\phi_{\pat}$ is straightforward from the definition. 
\OMIT{
Recall that our first assumption was that $\pat$ does not have variabeles on the edges. We show that this assumption is justified by showing that if there are variables on the edges, we can transform $\pat$ into $\pat'$ and change accordingly $\dlpath$ into $\dlpath'$.
To transform $\pat$ into $\pat'$ we replace $(i)$ each $\rightarrow$ with $()\rightarrow()\rightarrow()$  and $(ii)$ each $\overset{x}{\rightarrow}$ with $()\rightarrow(x)\rightarrow()$. The dataless path $\dlpath'$ is obtained from $\dlpath$ by inserteing a fresh node in between each two adjaecent nodes, e.g., if $\dlpath$ is $n_1 \rightarrow n_2 \rightarrow n_3$ then in $\dlpath'$ is   $n_1 \rightarrow n'_1 \rightarrow n_2 \rightarrow n'_2 \rightarrow n_3$ .
We can then obatain a formula $\phi_{\pat'}$  such that 
\[
\phi_{\pat'}(i'_s, i'_1, \ldots, i'_m, i'_t) \,\text{ if and only if }\, \exists p:\, (p,\mu')\in\sem{\pat'}_{\dlpath'}
\] 
for which the conditions of the lemma holds. 
We can reconstruct $\phi_{\pat}$ 
 if $\mu(x_j)$ is odd then $i'_j = \pos(\mu(x_j))$, and otherwise ($\mu(x_j) \in \Edgeset$) then $i_j=2k+1$ where $\mu(x_j)$ is the edge between node in position $k$ and $k+1$.  We can then construct $\phi_{\pat}$ from $\phi_{\pat'}$.
}
\end{proof}
\newcommand{\lmarker}{\triangleleft}
\newcommand{\rmarker}{\triangleright}
\OMIT{\begin{lemma}\label{lem:patrep}
    For every path pattern of the form $\pat^{*}$ without variables, there exists a two-way finite automaton 
    $A_{\pat}$  such that
    \[\{|w|\mid  \lmarker w \rmarker \in \lang{A_{\pat}}\} = \left\{\pos(n_t) - \pos(n_s) \, \middle| \begin{array}{l}
         (p,\emptyset)\in \sem{\pat^{*}}_{\dlpath}  \\
         n_s,n_s\,\text{ are first and last nodes in $p$, respectively}
    \end{array}  \right\}\]
\end{lemma}
}

\begin{lemma}
    For every data-less path $\dlpath$, and path pattern with output $\pat_{\return}$ with variables only on nodes, 
    there is a PA formula $\phi_{\pat}(x_1,\ldots, x_m)$ where   $\return \df\{x_1,\ldots, x_m\}$ such that 
    \[
\phi_{\pat_{\return}}(i_1,\ldots, i_m) \,\text{ if and only if }\,
  \mu \in\sem{\pat}_{\dlpath} 
    \]
    where  $i_j = \pos(\mu(x_j))$ for every $j$.
\end{lemma}
\begin{proof}
    Notice that since we are dealing with dataless paths as input, all elements in $\return$ are variables. Hence, the proof is implied directly from Lemma~\ref{lem:dlp_to_pa}. In particular, by existentially quantifying over variables that are omitted in $\return$. For example,\begin{align*}
    &\text{if}\, \phi_{\pat}(x_s,\bar{x}, \bar{y} , x_t) \,\text{then}\,\phi_{\pat_{ \bar y}}(\bar{y} ) = \exists x_s,\bar x,x_t:\, \phi_{\pat}(x_s,\bar{x}, \bar{y} , x_t) \\
     &\text{if}\, \phi_{\pat}(x_s,\bar{x}, \bar{z},\bar{y} , x_t) \,\text{then}\,\phi_{\pat_{ \bar{y}\bar{x}}}( \bar{y},\bar{x} ) = \exists x_s,\bar{z},x_t:\, \phi_{\pat}(x_s,\bar{x}, \bar{z},\bar{y} , x_t)  
      \end{align*}
\end{proof}

\begin{lemma}
  For every data-less path $\dlpath$, and Core GQL query $\query$,
    there is a PA formula $\phi_{\query}(x_1,\ldots, x_m)$ such that $\attr{\query} = \{x_1,\ldots, x_m\}$ and 
    \[
\phi_{\query}(i_1,\ldots, i_m) \,\text{ if and only if }\,
  \mu \in\sem{\query}_{\dlpath} 
    \]
    where  $i_j = \pos(\mu(x_j))$ for every $j$.
\end{lemma}
\begin{proof}
    The proof is by induction on the structure of $\query$.
    The base case of path pattern with output is covered by the previous lemma.
    For the induction step, we distinguish between the form of $\query$.
    \begin{itemize}
        \item
        If $\query \df\pat_{\mathbf A}(\query')$ then
        $\phi_{\query}(\bar z) \df \exists \bar y \, 
        \phi_{\query'}(x_1,\ldots, x_m)$ where $\bar y = \mathbf A \setminus \{x_1,\ldots, x_m\} $ and $\bar z = \mathbf A \cap \{x_1,\ldots, x_m\} $.
        \item 
        If $\query \df \sigma_{\theta}(\query')$ then $\phi_{\query}(\bar x) \df \phi_{\query'}(\bar x)$. Notice that the correctness here follows from the fact the input is a data-less path.
        \item 
        If $\query \df \rho_{A\rightarrow A'}(\query')$ then $\phi_{\query}(\bar{x}, A, \bar{y} ) \df 
        \phi_{\query'}(\bar{x}, A', \bar{y} ) 
        $.
        \item
        If $\query \df \query_1 \times \query_2$ then $\phi_{\query}(\bar{x}) \df \phi_{\query_1}(\bar{z}) \wedge \phi_{\query_2}(\bar{y})$ with $\bar{x}$ the union of $\bar{z}$ and $\bar{y}$.
        \item 
        If $\query \df \query_1 \cup \query_2$ then $\phi_{\query}(\bar{x}) \df \phi_{\query_1}(\bar{x}) \vee \phi_{\query_2}(\bar{x})$.
        \item 
          If $\query \df \query_1 \cap \query_2$ then $\phi_{\query}(\bar{x}) \df \phi_{\query_1}(\bar{x}) \wedge \phi_{\query_2}(\bar{x})$.
                \item 
          If $\query \df \query_1 \setminus \query_2$ then $\phi_{\query}(\bar{x}) \df \phi_{\query_1}(\bar{x}) \wedge \neg \phi_{\query_2}(\bar{x})$.
    \end{itemize}
    Showing that the condition holds is straightforward from the definition.
\end{proof}

\begin{proof}[Proof of Theorem~\ref{thm:2_pow_n}]

    Let us assume by contradiction that there is a Boolean GQL query $\query$ for which $\sem{\query}_{\dlpath} \ne \emptyset $ if and only if $\len(\dlpath)=2^n$, $n\in \mathbb{N}$.
\OMIT{
For dataless path, there is no way to distinguish when a path starts and where it ends. We can show that for dataless path $\dlpath$ ...

So we can assume that the dataless path is a bit extended by adding labels $\mathsf{first}$ and $\mathsf{last}$ to the first and last nodes of the path, respectively.  NO NEEDDDDD

All previous results can be extended, in particular, the previous lemma.
Now, it holds that  $\sem{\query}_\dlpath \ne \emptyset $ iff $\sem{ \pat_{x}(\query \times \query')}_\dlpath \ne \emptyset$ with 
$\query' = (:\mathsf{first}) \rightarrow^{*}(x:\mathsf{last}) $.
But then in the semantics $\sem{ \pat_{x}(\query \times \query')}_\dlpath$ there is a mapping $\mu$ with $\mu(x) = n_k$ where $n_k$ is the last node of $\dlpath$. This is true for every $\dlpath$ of length $2^n,n\in\mathbb{N}$ and since $\{ 2^n \mid n\in \mathbb N \}$ is not semi-linear, we reach the desired contradiction. 
}

Let us define a new query $\query'$ that binds $x$ to the last node of a data-less path. We set $\query' \df \pat_1 \setminus \pat_2$ where $\pat_1 \df  ()\rightarrow^{*}(x)$ and $\pat_2 \df ()\rightarrow^{*}(x)\rightarrow ()$.
It holds that $ \sem{\query'}_{\dlpath} \df \{\mu \mid \dom{\mu}= \{x\}, \mu(x) = n_k \}$ where $n_k$ is the last node in $\dlpath$.
Now let us consider $\query \times \query'$. It can be easily shown that $\sem{\pat_x(\query \times \query')}_{\dlpath} \ne \emptyset$ iff $\sem{\query }_{\dlpath} \ne \emptyset$. 
By previous lemma, $\sem{\pat_x(\query \times \query')}_{\dlpath} $ is expressible by a PA formula.
Nevertheless, by our assumption
$\sem{\pat_x(\query \times \query')}_{\dlpath} $ is exactly the set $\{2^n \mid n\in \mathbb{N} \}$ which leads to the desired contradiction. 
\end{proof}

\subsection{Datalog on Graphs}

\subsubsection*{Datalog}

Let $\mathcal S$ be a schema.
\emph{An atom} over 
$\mathcal{S}$ is an element of the form 
$S(x_1,\ldots, x_k)$ where $S\in\mathcal{S}$, $\arity{S}=k$ and $x_1,\ldots x_k$ is a sequence of variables with no repetitions.  
 \emph{A Datalog program} is a tuple $(\edb,\idb,\dlrules,\dloutrule)$ where $\edb$ and $ \idb$, namely \emph{the IDB and EDB signatures}, respectively, are disjoint schemas, $\dloutrule$ is a designated relation symbol in $\idb$, and $\dlrules$ is a set of \emph{rules} $\rho$ of the form 
\[
\phi\leftarrow \phi_1 , \ldots, \phi_m
\]
where $\phi$ is an atom over $\idb$, and each $\phi_i$ is an atom over $\idb \cup \edb$.
We restrict further the rules by allowing only cases in which each variable in the  \emph{head} $\phi$ of $\rho$ occurs also in its \emph{body}, namely $\phi_1, \ldots, \phi_m$.

\emph{The semantics $\sem{P}_{\db}$ of a Datalog program} $P \df  (\edb,\idb,\dlrules,\dloutrule)$ over a database $\db$ is defined in the standard way~\cite{} as we recall now:
Given a database 
$\db$ over a schema 
$\localsch$, a partial mapping 
$\mu:\Vars \rightarrow \Univ $ \emph{satisfies an atom $\phi(x_1,\ldots,x_k)$ (w.r.t. $\db$)} 
if $\phi \in \localsch$ and 
$(\mu(x_1),\ldots,\mu(x_k))\in \db({\phi})$.
A database $\db'$ satisfies a rule $\rho$ if 
every partial mapping $\mu$ that satisfies each of $\rho$'s body atoms w.r.t. $\db'$, also satisfies its head atom. 
A database $\db'$ satisfies a set $\Phi$ of rules if it satisfies every $\rho\in \Phi$.
We define the semantics $ \sem{P}_{\db} \df \db'$ such that (1) $\db'$ satisfies $\dlrules$ and (2)  every $\db''$ that satisfies $\dlrules$ has the following property: every relation symbol $R \in \idb \cup \edb $, the set $\db'(R) \subseteq \db''(R)$.
We say that $P$ is Boolean if $\arity{\dloutrule} = 0$, and in this case, we say it outputs $\true$ if $\dloutrule() \in\sem{P}_{\db}$.

\paragraph*{Datalog over property graphs}
We view property graphs $
    G \df \langle N, E, \lbl, \src, \tgt, \prop\rangle$ as relational structures $\db_G$  over the relational schema consisting of relation symbols $N, E, \lbl, \src, \tgt, \prop$ with the straightforward interpretation.

\newcommand{\idbname}[1]{\mathsf{#1}}
\begin{example}
The Datalog program defined by the rules: 
    \begin{align*}
        \idbname{eqLen}(x,y,z,w)  
        &\leftarrow 
        E(x,y),E(z,w)\\
        \idbname{eqLen}(x,y,z,w)  
        &\leftarrow
         \idbname{eqLen}(x,y',z,w'), 
        E(y',y),E(w'w)
    \end{align*} 
    extracts from a bounded data-less path bindings $\mu$ of $x,y,z,w$ such that the number of edges between $\mu(x)$ and $\mu(y)$ equals to that between $\mu(z)$ and $\mu(w)$.
\end{example}

\begin{example}
    The Datalog program given by the rules
\begin{align*}
    \idbname{len}_{2^n}(x,y) &\leftarrow 
   E(x,z),E(z,y)
    \\
    \idbname{len}_{2^n}(x,y) &\leftarrow \idbname{len}_{2^n}(x,z), \idbname{len}_{2^n}(w,y),
    \idbname{eqLen}(x,z,w,y)
    \\
    \dloutrule() &\leftarrow \idbname{len}_{2^n}(x,y), \lbl(x,\mathsf{first}),  \lbl(y,\mathsf{last})
\end{align*}
outputs $\true$ if and only if the length of the input bounded data-less path is $2^n$ for some $n\in \mathbb{N}$.
\end{example}

Using a straightforward induction, we can show that:
\begin{lemma}\label{lem:dlquotient}
For every Datalog program $P$  
and data-less paths $p,p$', if $p\sim p'$ then 
$\sem{P}_{\db_{p'}} = \sem{P}_{\db_p}$.
\end{lemma}
We denote $\db_p$ by the $\db_n$ where $n$ is the length of $p$.
We can view Boolean Datalog Programs as queries on the quotient space $P/\sim$ and 
define 
$$\len(P) \df \{n \mid \sem{P}_{\db_{n}} = \true \}. $$
We can conclude:
\begin{corollary}\label{cor:dl2pown}
Datalog over property graphs is strictly more expressive than GQL.
There is a Datalog over property graphs program such that for every data-less path $\mathbf{p}$ the following holds:
\[
\sem{\query}_\mathbf{p} = \true \text{ if and only if }\, 
\len(\mathbf{p})=2^n,n\in \mathbb{N}
\]
\end{corollary}

\subsection{Proof of Theorem~\ref{main-gql-thm}}

The proof follows directly from Theorem~\ref{thm:2_pow_n} and Corollary~\ref{cor:dl2pown}.